\documentclass[a4paper,11pt,reqno]{article}
\usepackage{etoolbox,xstring}
\providebool{csub}
\providebool{jour}
\setbool{csub}{false}
\setbool{jour}{true}

\usepackage[utf8]{inputenc}
\usepackage{amsmath, amssymb, amsfonts, amsthm}
\usepackage{bm}
\renewcommand{\mathbf}[1]{\boldsymbol{#1}}

\usepackage{nicefrac}

\usepackage{libertine}
\usepackage{microtype}
\usepackage[T1]{fontenc}

\usepackage{booktabs,siunitx}

\usepackage{comment}
\usepackage{marginnote}
\usepackage{xcolor}
\usepackage{framed}

\usepackage[noadjust]{cite}
\newcommand{\arxiv}[1]{\href{http://arxiv.org/abs/#1}{arxiv:#1}}
\usepackage[colorlinks=true,citecolor=blue,bookmarksnumbered=true,hyperfootnotes=false]{hyperref}

\newtheorem{theorem}{Theorem}[section]

\newtheorem{condition}[theorem]{Condition}

\newtheorem{lemma}[theorem]{Lemma}
\newtheorem{corollary}[theorem]{Corollary}
\theoremstyle{definition}
\newtheorem{definition}{Definition}[section]
\newtheorem{remark}{Remark}[section]

\newcommand{\abs}[1]{\ensuremath{\left|#1\right|}}
\newcommand{\norm}[2][]{\ensuremath{\Vert #2 \Vert_{#1}}}

\newcommand{\diff}[2]{\frac{\text{d}#1}{\text{d}#2}}
\newcommand{\pdiff}[2]{\frac{\partial #1}{\partial #2}}

\renewcommand{\Pr}[2][]{\ensuremath{\mathbb{P}_{#1}\insq{#2}}}
\newcommand{\E}[2][]{\ensuremath{\mathbb{E}_{#1}\insq{#2}}}

\newcommand{\ina}[1]{\left<#1\right>}
\newcommand{\inb}[1]{\left\{#1\right\}}
\newcommand{\inp}[1]{\left(#1\right)}
\newcommand{\insq}[1]{\left[#1\right]}

\makeatletter
\newcommand*{\defeq}{\mathrel{\rlap{%
                     \raisebox{0.3ex}{$\m@th\cdot$}}%
                     \raisebox{-0.3ex}{$\m@th\cdot$}}%
                    =}
\makeatother

\newcommand{\Z}[0]{\ensuremath{\mathbb{Z}}}
\newcommand{\R}[0]{\ensuremath{\mathbb{R}}}

\renewcommand{\vec}[1]{\bm{#1}}

\newcommand{\Ts}[1]{\ensuremath{T_{SAW}\inp{#1}}}

\newcommand{\G}[0]{\mathcal{G}}
\newcommand{\br}[1]{\mathrm{br}\inp{#1}}

\newcommand{\family}[1]{{\mathcal{#1}}}

\usepackage[paperwidth=8.5in,paperheight=11.0in,margin=1.0in]{geometry}

\begin{document}
\title{Spatial mixing and the connective constant: Optimal bounds}
\newcommand{\sincgrant}{NSF grant CCF-1016896}
\newcommand{\calgrant}{NSF grant CCF-1319745}
{\renewcommand{\thefootnote}{} \footnotetext{Weaker versions of some
    of the results in this paper appeared in  \emph{Proceedings of the
    IEEE Symposium on the Foundations of Computer Science
      (FOCS)}, 2013, pp. 300-309~\cite{sinclair13:_spatial}. This
    version strengthens the main result (Theorem 1.3) of
    \cite{sinclair13:_spatial} to obtain an optimal setting of
    the parameters, and adds new results for the monomer-dimer
    model.}  }

  \author{
  {Alistair Sinclair}\thanks{Alistair Sinclair, Computer
    Science Division, UC Berkeley. Email:
    \texttt{sinclair@cs.berkeley.edu}. Supported in part by
    \sincgrant~and by the Simons Institute for the Theory of
    Computing.} \and
    {Piyush Srivastava}\thanks{Piyush Srivastava, Center for
      the Mathematics of Information, Caltech. Email:
      \texttt{piyushsriva@gmail.com}. Supported by \calgrant.  This
      work was done while this author was a graduate student at UC
      Berkeley and was supported by \sincgrant.} \and
    {Daniel {\v{S}tefankovi\v{c}}}\thanks{Daniel
      {\v{S}tefankovi\v{c}}, Department of Computer Science, University
      of Rochester. Email: \texttt{stefanko@cs.rochester.edu}.
      Supported in part by \sincgrant. Part of this work was done
      while this author was visiting the Simons Institute for the Theory
      of Computing.} \and
    {Yitong Yin}\thanks{Yitong Yin, State Key Laboratory for
      Novel Software Technology, Nanjing University, China. Email:
      \texttt{yinyt@nju.edu.cn}. Supported by NSFC grants 61272081 and
      61321491. Part of this work was done while this author was
      visiting UC Berkeley.}} \date{}
  \maketitle
{
  \begin{abstract}
  We study the problem of deterministic approximate counting of
  matchings and independent sets in graphs of bounded \emph{connective
    constant}.  More generally, we consider the problem of evaluating
  the partition functions of the \emph{monomer-dimer model} (which is
  defined as a weighted sum over all matchings where each matching is
  given a weight $\gamma^{|V| - 2 |M|}$ in terms of a fixed parameter
  $\gamma$ called the \emph{monomer activity}) and the \emph{hard core
    model} (which is defined as a weighted sum over all independent
  sets where an independent set $I$ is given a weight $\lambda^{|I|}$
  in terms of a fixed parameter $\lambda$ called the \emph{vertex
    activity}).  The \emph{connective constant} is a natural measure of
  the average degree of a graph which has been studied extensively in
  combinatorics and mathematical physics, and can be bounded by a
  constant even for certain unbounded degree graphs such as those
  sampled from the sparse Erdős–Rényi model $\mathcal{G}(n, d/n)$.

  Our main technical contribution is to prove the best possible rates
  of decay of correlations in the natural probability distributions
  induced by both the hard core model and the monomer-dimer model in
  graphs with a given bound on the connective constant.  These results
  on decay of correlations are obtained using a new framework based on
  the so-called \emph{message} approach that has been extensively used
  recently to prove such results for bounded degree graphs.  We then
  use these optimal decay of correlations results to obtain FPTASs for
  the two problems on graphs of bounded connective constant.

  In particular, for the monomer-dimer model, we give a deterministic
  FPTAS for the partition function on all graphs of bounded connective
  constant for any given value of the monomer activity.  The best
  previously known deterministic algorithm was due to Bayati,
  Gamarnik, Katz, Nair and Tetali~[STOC 2007], and gave the same
  runtime guarantees as our results but only for the case of bounded
  degree graphs.  For the hard core model, we give an FPTAS for graphs
  of connective constant $\Delta$ whenever the vertex activity
  $\lambda < \lambda_c(\Delta)$, where $\lambda_c(\Delta) \defeq
  \frac{\Delta^\Delta}{(\Delta - 1)^{\Delta + 1}}$; this result is
  optimal in the sense that an FPTAS for any $\lambda >
  \lambda_c(\Delta)$ would imply that NP=RP [Sly, FOCS 2010].  The
  previous best known result in this direction was a recent paper by a
  subset of the current authors~[FOCS 2013], where the result was
  established under the sub-optimal condition $\lambda <
  \lambda_c(\Delta + 1)$.

  Our techniques also allow us to improve upon known bounds for decay
  of correlations for the hard core model on various regular lattices,
  including those obtained by Restrepo, Shin, Vigoda and Tetali~[FOCS
  11] for the special case of $\mathbb{Z}^2$ using sophisticated
  numerically intensive methods tailored to that special case.
\end{abstract}
}
\thispagestyle{empty}

\newpage
\setcounter{page}{1}
\section{Introduction}
\label{sec:introduction}

\subsection{Background}
\label{sec:background}

This paper studies the problem of approximately counting independent
sets and matchings in sparse graphs.  We consider these problems
within the more general formalism of \emph{spin systems}.  In this
setting, one first defines a natural probability distribution over
\emph{configurations} (e.g., independent sets or matchings) in terms of
local \emph{interactions}.  The counting problem then corresponds to
computing the normalization constant, known as the \emph{partition
  function} in the statistical physics literature.  The partition
function can also be seen as a generating function of the
combinatorial structures being considered and is an interesting graph
polynomial in its own right.

The first model we consider is the so called \emph{hard core model},
which is defined as follows.  We start with a graph $G = (V, E)$, and
specify a \emph{vertex activity} or \emph{fugacity} parameter $\lambda
> 0$.  The configurations of the hard core model are the independent
sets of the graph, and the model assigns a \emph{weight} $w(I) =
\lambda^{|I|}$ to each independent set $I$ in $G$.
The weights in turn determine a natural probability distribution
$\mu(I) = \frac{1}{Z}w(I)$
over the independent sets known as the \emph{Gibbs
    distribution}. Here,
\[
Z = Z(\lambda) \defeq \sum_{I: \text{independent set}} w(I)
\]
is the \emph{partition function}.  Clearly, the problem of counting
independent sets is the special case $\lambda = 1$.

Our next model is the \emph{monomer-dimer model}, which has as its
configurations all matchings of a given graph $G = (V, E)$. For a
specified \emph{dimer activity} $\gamma > 0$, the model assigns a
weight $w(M) = \gamma^{|M|}$ to each matching $M$ of the graph.  As
before, the weights define the Gibbs distribution $\mu(M) =
\frac{1}{Z} w(M)$ over matchings, where
\[
Z = Z(\gamma) \defeq \sum_{M: \text{matching}} w(M)
\]
is the partition function.  The problem of counting matchings
again corresponds to the special case $\gamma = 1$.

The problem of approximating the partition function has received much
attention, both as a natural generalization of counting and because of
its connections to sampling.\footnote{For the large class of
  self-reducible problems, it can be shown that approximating the
  partition function is polynomial-time equivalent to approximate
  sampling from the Gibbs distribution~\cite{jervalvaz86}.}
Recent progress in relating
the complexity of approximating the partition function to phase
transitions, which we now describe, has provided further impetus to
this line of research.

The first such result was due to
Weitz~\cite{Weitz06CountUptoThreshold}, who exploited the properties
of the Gibbs
measure of the hard core model on the infinite $d$-ary tree.  It was
well known that this model exhibits the following phase transition:
there exists a \emph{critical activity} $\lambda_c(d)$ such that the
total variation distance between the marginal probability
distributions induced at the root of the tree by \emph{any} two
fixings of the independent set on all the vertices at distance $\ell$
from the root decays exponentially in $\ell$ when $\lambda <
\lambda_c(d) \defeq \frac{d^d}{(d-1)^{d+1}}$, but remains bounded away
from $0$ even as $\ell\rightarrow \infty$ when $\lambda >
\lambda_c(d)$. (The former condition is also referred to as
\emph{correlation decay}, since the correlation between the
configuration at the root of the tree and a fixed configuration at
distance $\ell$ from the root decays exponentially in $\ell$; it is
also called \emph{spatial mixing}.)
Weitz showed that for all $\lambda <
\lambda_c(d)$ (i.e., in the regime where correlation decay holds on
the $d$-ary tree), there exists a deterministic FPTAS for the
partition function of the hard core model on all graphs of degree at
most $d + 1$. (Note that the condition on $\lambda$ is only in terms
of the $d$-ary tree, while the FPTAS applies to all graphs.)  This
connection to phase transitions was further strengthened by
Sly~\cite{Sly2010CompTransition} (see
also~\cite{sly12,Vigoda-hard-core-11}), who showed that an FPRAS for
the partition function of the hard core model with $\lambda >
\lambda_c(d)$ on graphs of degree $d+1$ would imply NP = RP.

In addition to establishing a close connection between the complexity
of a natural computational problem and an associated phase transition, Weitz's
algorithm had the further interesting feature of not being based on
Markov chain Monte Carlo (MCMC) methods; rather, it used a deterministic
procedure based on proving that decay of correlations on the
$d$-ary tree implies decay of correlations on all graphs of degree at
most $d+1$.  To date, no MCMC algorithms are known for the
approximation of the partition function of the hard core model on
graphs of degree at most $d+1$ which run in polynomial time for all
$\lambda < \lambda_c(d)$.

Weitz's algorithm led to an exploration of his approach for other
problems too.  For example, in the case of the monomer-dimer
model---unlike that of the hard core model---there does exists a
\emph{randomized} polynomial time algorithm (based on MCMC) for
approximating the partition function which works for every $\gamma >
0$, without any bounds on the degree of the
graph~\cite{jerrum_approximating_1989}.  However, finding a
deterministic algorithm for the problem remains open.  Bayati,
Gamarnik, Katz, Nair and Tetali~\cite{bayati_simple_2007} made
progress on this question by showing that Weitz's approach could be
used to derive a deterministic algorithm that runs in polynomial time
for bounded degree graphs, and is sub-exponential on general graphs.

The algorithms of both Weitz and Bayati \emph{et al.} are therefore
polynomial time only on bounded degree graphs, and in particular, for
a given value of the parameter $\lambda$ (or $\gamma$ in the case of the
monomer-dimer model) the running time of these algorithms on graphs of
maximum degree $d+1$ depends upon the rate of decay of correlations on
the infinite $d$-ary tree.  Further, these results are obtained by
showing that decay of correlations on the $d$-ary tree implies a
similar decay on all graphs of maximum degree $d+1$.

There are two important shortcomings of such results.  First, in
statistical physics one is often interested in special classes of
graphs such as regular lattices.  One can reasonably expect that the
rate of decay of correlations on such graphs should be better than
that predicted by their maximum degree.  Second, these results have
no non-trivial consequences even in very special classes of sparse
unbounded degree graphs, such as graphs drawn from the
Erd\H{o}s-R\'{e}nyi model $\mathcal{G}(n, d/n)$ for constant $d$.

This state of affairs leads to the following natural question: is
there a finer notion of degree that can be used in these results in
place of the maximum degree?  Progress in this direction was made
recently for
the case of the hard core model in~\cite{sinclair13:_spatial}, where
it was shown that one can get decay of correlation results in terms of
the \emph{connective constant,} a natural and well-studied notion of
average degree.  The connective constant of a regular lattice of
degree $d+1$ is typically substantially less than $d$; and it is
bounded even in the case of sparse random graphs such as those drawn
from $\mathcal{G}(n, d/n)$, which have unbounded maximum degree. By
analogy with the bounded degree case, one might hope to get
correlation decay on graphs with connective constant at most $\Delta$
for all $\lambda < \lambda_c(\Delta)$.  In~\cite{sinclair13:_spatial},
such a result was proven under the stronger condition $\lambda <
\lambda_c(\Delta + 1)$.  The latter bound is tight asymptotically as
$\Delta \rightarrow \infty$ (because
$\lambda_c(\Delta+1)/\lambda_c(\Delta) \rightarrow 1$ as $\Delta
\rightarrow \infty$), but is sub-optimal in the important case of small
$\Delta$.

\subsection{Contributions}
\label{sec:our-results}

In this paper, we show that one can indeed replace the maximum degree
by the connective constant in the results of both
Weitz~\cite{Weitz06CountUptoThreshold} and Bayati \emph{et
  al.}~\cite{bayati_simple_2007}.  In particular, we show that for
both the hard core and the monomer-dimer models, decay of correlations
on the $d$-ary tree determines the rate of decay of correlations---as
well as the complexity of deterministically approximating the
partition function---in all graphs of \emph{connective constant} at
most $d$, without any dependence on the maximum degree.  The specific
notion of decay of correlations that we establish is known in the
literature as \emph{strong spatial
  mixing}~\cite{Weitz06CountUptoThreshold, goldberg_strong_2005,
  martinelli_approach_1994, martinelli_approach_1994-1}, and
stipulates that the correlation between the state of a vertex $v$ and
another set $S$ of vertices at distance $\ell$ from $v$ should decay
exponentially in $\ell$ even when one is allowed to fix the state of
vertices close to $v$ to arbitrary values (see
Section~\ref{sec:strong-spat-mixing} for a precise definition).  Prior to
the role it played in the design of deterministic approximate counting
algorithms in Weitz's work~\cite{Weitz06CountUptoThreshold}, strong
spatial mixing was already a widely studied notion in computer science
and mathematical physics for its utility in analyzing the mixing time
of Markov chains~\cite{martinelli_approach_1994,
  martinelli_approach_1994-1, goldberg_strong_2005}, and hence
an improved understanding of conditions under which it holds is of
interest in its own right.

We now give an informal description of the connective
constant~\cite{hammersley_percolation_1957,madras96:_self_avoid_walk};
see Section~\ref{sec:connective-constant} for precise definitions.
Given a graph $G$ and a vertex $v$ in $G$, let $N(v, \ell)$ denote the
number of self avoiding walks in $G$ of length $\ell$ starting at $v$.
A graph family $\mathcal{F}$ is said to have connective constant
$\Delta$ if for all graphs in $\mathcal{F}$, the number of
self-avoiding walks of length at most $\ell$ for large $\ell$ grows as
$\Delta^\ell$, i.e., if $\ell^{-1}\log \sum_{i=1}^\ell N(v, i)\sim
\log \Delta$ (the definition can be applied to both finite and
infinite graphs; see Section~\ref{sec:connective-constant}). Note that
in the special case graphs of maximum degree $d+1$, the connective
constant is at most $d$.  It can, however, be much lower that this
crude bound: for any $\epsilon > 0$, the connective constant of graphs
drawn from $G(n, d/n)$ is at most $d (1 + \epsilon)$ with high
probability (w.h.p.)~(see, e.g., \cite{sinclair13:_spatial}), even
though their maximum degree is $\Omega\inp{\frac{\log n}{\log \log n}}$
w.h.p.

Our first main result can now be stated as follows.
\begin{theorem}[\textbf{Main, Hard core model}]
  \label{thm:main-hard-core}
  Let $\mathcal{G}$ be a family of
  finite graphs of connective constant at most $\Delta$, and let
  $\lambda$ be such that $\lambda < \lambda_c(\Delta)$.  Then there is
  an FPTAS for the partition function of the hard core model with
  vertex activity $\lambda$ for all graphs in $\mathcal{G}$.  Further,
  even if $\mathcal{G}$ contains locally finite infinite graphs, the
  model exhibits strong spatial mixing on all graphs in $\mathcal{G}$.
\end{theorem}
\begin{remark}
  \label{rem:hard-core}
  In \cite{sinclair13:_spatial}, the above result was proved under the
  stronger hypothesis $\lambda < \lambda_c(\Delta + 1)$.  The above
  result therefore subsumes the main results of
  ~\cite{sinclair13:_spatial}.  It is also optimal in the following
  sense: there cannot be an FPRAS for graphs of connective constant at
  most $\Delta$ which works for $\lambda > \lambda_c(\Delta)$, unless
  NP = RP.  This follows immediately from the hardness results for the
  partition function of the hard core model on bounded degree
  graphs~\cite{sly12,Sly2010CompTransition} since graphs of degree at
  most $d+1$ have connective constant at most $d$.
\end{remark}
An immediate corollary of Theorem~\ref{thm:main-hard-core} is the
following.
\begin{corollary}
  \label{cor:hard-core-gndn}
  Let $\lambda < \lambda_c(d)$.  Then, there is an algorithm for
  approximating the partition function of graphs drawn from
  $\mathcal{G}(n, d/n)$ up to a factor of $(1 \pm \epsilon)$ which,
  with high probability over the random choice of the graph, runs in
  time polynomial in $n$ and $1/\epsilon$.
\end{corollary}
Similar results for $\G(n,d/n)$ have appeared in the literature in the
context of rapid mixing of Glauber dynamics for the ferromagnetic
Ising model~\cite{mossel_exact_2013}, and also for the hard core
model~\cite{mossel_gibbs_2010,efthymiou13:_mcmc_g}.  Although the
authors of~\cite{mossel_gibbs_2010} do not supply an explicit range
of~$\lambda$ for which their rapid mixing results hold, an examination
of their proofs suggests that necessarily $\lambda<O\inp{1/d^2}$.
Similarly, the results of \cite{efthymiou13:_mcmc_g} hold when
$\lambda < 1/(2d)$.  In contrast, our bound approaches (and is always
better than) the conjectured optimal value $e/d$. Further, unlike
ours, the results of~\cite{mossel_gibbs_2010,
  mossel_exact_2013,efthymiou13:_mcmc_g} are restricted to $\G(n,
d/n)$ and certain other classes of sparse graphs.

A second consequence of Theorem~\ref {thm:main-hard-core} is a further
improvement upon the spatial mixing bounds obtained
in~\cite{sinclair13:_spatial} for various lattices, as shown in
Table~\ref {fig:1}.  For each lattice, the table shows the best known
upper bound for the connective constant and the strong spatial mixing
(SSM) bounds we obtain using these values in
Theorem~\ref{thm:main-hard-core}.  In the table, a value~$\alpha$
in the ``$\lambda$'' column means that SSM is shown to hold for the
appropriate lattice whenever $\lambda \leq \alpha$.  As expected,
improvements over our previous results in~\cite{sinclair13:_spatial}
are the most pronounced for lattices with smaller maximum degree.

The table shows that except in the case of the 2D integer lattice
$\mathbb{Z}^2$, our general result immediately gives improvements on
the best known SSM
bounds for all lattices using only previously known estimates of the
connective constant.  Not unexpectedly, our bound for $\mathbb{Z}^2$
using the connective constant as a black-box still improves upon
Weitz's bound but falls short of the bounds obtained by Restrepo
\emph{et al.}~\cite{restrepo11:_improv_mixin_condit_grid_count} and
Vera \emph{et al.}~\cite{vera13:_improv_bound_phase_trans_hard} using
numerically intensive methods tailored to this special case.
However, as we noted in~\cite{sinclair13:_spatial}, any improvement in
the bound on the connective constant would immediately yield an
improvement in our SSM bound. %
Indeed, in Appendix~\ref{sec:descr-numer-results},
we use a tighter analysis of the connective constant of a suitably
constructed self-avoiding walk tree of $\mathbb{Z}^2$ to show that SSM
holds on this lattice whenever $\lambda < 2.538$, which improves upon
the specialized bound $\lambda < 2.48$, obtained in the
papers~\cite{vera13:_improv_bound_phase_trans_hard,
restrepo11:_improv_mixin_condit_grid_count}.  We note that this
improvement would not be possible using only our earlier results
in~\cite{sinclair13:_spatial}.

\newbool{table}
\setbool{table}{true}
\ifbool{table}{
\begin{table}[t]
\begin{minipage}[t]{\textwidth}
  \centering
  \renewcommand{\thefootnote}{$\star$}
  \begin{tabular}[h]{c
      S[table-format=2]
      S[table-format=0.4]
      S[table-format=3.5]
      l
      c
      l%
    }
    \toprule
    &\multicolumn{1}{c}{Max.}&\multicolumn{1}{c}{Previous SSM bound}&
    \multicolumn{2}{c}{Connective Constant}&
    SSM bound in~\cite{sinclair13:_spatial}&
    \multicolumn{1}{c}{Our SSM bound}\\
    \cmidrule(r){3-3} \cmidrule(r){4-5} \cmidrule(r){6-6} \cmidrule(r){7-7}
    Lattice &{degree}
    & {$\lambda$}
    & \multicolumn{2}{c}{$\Delta$} & \multicolumn{1}{c}{$\lambda$}
    & \multicolumn{1}{c}{$\lambda$} \\

    \midrule

    $\mathbb{T}$  &  6  &  0.762 \cite{Weitz06CountUptoThreshold} &
    4.251 419 & \cite{alm_upper_2005}  &
    0.937 & 0.961  \\

    $\mathbb{H}$  &  3  &  4.0 \cite{Weitz06CountUptoThreshold} &
    1.847 760 & \cite{duminil-copin_connective_2012}  &
    4.706 & 4.976  \\

    $\Z^2$  &  4  &  2.48
    \cite{restrepo11:_improv_mixin_condit_grid_count,vera13:_improv_bound_phase_trans_hard}
    &  2.679 193 & \cite{poenitz00:_improv_z} &
    2.007 &  2.082 (2.538\footnotemark)  \\

    $\Z^3$  &  6  & 0.762 \cite{Weitz06CountUptoThreshold} &  4.7387 &
    \cite{poenitz00:_improv_z}  &
    0.816 & 0.822  \\

    $\Z^4$   &  8  &  0.490 \cite{Weitz06CountUptoThreshold} &  6.804
    0 & \cite{poenitz00:_improv_z}  &
    0.506 & 0.508  \\

    $\Z^5$  &  10  &  0.360 \cite{Weitz06CountUptoThreshold} &  8.860
    2 & \cite{poenitz00:_improv_z}  &
    0.367 & 0.367  \\

    $\Z^6$  &  12  &    0.285 \cite{Weitz06CountUptoThreshold} &
    10.888 6 & \cite{weisstein:_self_avoid_walk_connec_const} &
    0.288 & 0.288  \\

    \bottomrule
  \end{tabular}
  \footnotetext{$^\star$ See Appendix~\ref{sec:descr-numer-results}
for a description of how this improved bound is obtained.}
\end{minipage}
  \caption{Strong spatial mixing bounds for various lattices.  ($\mathbb{Z}^D$
    is the $D$-dimensional Cartesian lattice; $\mathbb{T}$ and $\mathbb{H}$
    denote the triangular and honeycomb lattices respectively.)}
  \label{fig:1}
\end{table}
}{
\begin{table}[t]
\begin{minipage}[t]{\textwidth}
  \centering
  \renewcommand{\thefootnote}{$\star$}
  \begin{tabular}[h]{l
      S[table-format=5]
      S[table-format=0.4]
      S[table-format=2.6]
      c
      l
    }
    \toprule
    &&\multicolumn{1}{c}{Previous best SSM bound}&\multicolumn{2}{c}{Connective Constant}&\multicolumn{1}{c}{Our SSM bound}\\
    \cmidrule(r){3-3} \cmidrule(r){4-5} \cmidrule(r){6-6}
    Lattice &{Max. degree}
    & {$\lambda$}
    & \multicolumn{2}{c}{$\Delta$} & \multicolumn{1}{c}{$\lambda$}  \\

    \midrule

    $\mathbb{T}$  &  6  &  0.762 \cite{Weitz06CountUptoThreshold} &  4.251 419 & \cite{alm_upper_2005}  &  0.961  \\

    $\mathbb{H}$  &  3  &  4.0 \cite{Weitz06CountUptoThreshold} &  1.847 760 & \cite{duminil-copin_connective_2012}  &  4.976  \\

    $\Z^2$  &  4  &  2.48
    \cite{restrepo11:_improv_mixin_condit_grid_count,vera13:_improv_bound_phase_trans_hard}
    &  2.679 193 & \cite{poenitz00:_improv_z}  &  2.082
    (2.538\footnotemark)  \\

    $\Z^3$  &  6  & 0.762 \cite{Weitz06CountUptoThreshold} &  4.7387 & \cite{poenitz00:_improv_z}  &  0.822  \\

    $\Z^4$   &  8  &  0.490 \cite{Weitz06CountUptoThreshold} &  6.804 0 & \cite{poenitz00:_improv_z}  &   0.508  \\

    $\Z^5$  &  10  &  0.360 \cite{Weitz06CountUptoThreshold} &  8.860 2 & \cite{poenitz00:_improv_z}  &  0.367  \\

    $\Z^6$  &  12  &    0.285 \cite{Weitz06CountUptoThreshold} &  10.888 6 & \cite{weisstein:_self_avoid_walk_connec_const}  &  0.288  \\

    \bottomrule
  \end{tabular}
  \footnotetext{$^\star$ See Appendix~\ref{sec:descr-numer-results}
for a description of how this improved bound is obtained.}
\end{minipage}
  \caption{Strong spatial mixing bounds for various lattices.  ($\mathbb{Z}^D$
    is the $D$-dimensional Cartesian lattice; $\mathbb{T}$ and $\mathbb{H}$
    denote the triangular and honeycomb lattices respectively.)}
  \label{fig:1}
\end{table}
}
We also apply our techniques to the study of uniqueness of the Gibbs
measure of the hard core model on general trees.  Uniqueness is a
weaker notion than spatial mixing, requiring correlations to decay to
zero with distance but not necessarily at an exponential rate (see
Section~\ref{sec:branch-numb-uniq} for a formal definition).  We
relate the phenomenon of the uniqueness of the Gibbs measure of the
hard core model on a general tree to the \emph{branching factor} of
the tree, another natural notion of average arity that has appeared in
the study of uniqueness of Gibbs measure for models such as the
ferromagnetic Ising model~\cite{lyons_ising_1989}.  The details of
these results can be found in Section~\ref{sec:branch-numb-uniq}.  %

Our second main result concerns the monomer-dimer model.
\begin{theorem}[\textbf{Main, Monomer-dimer model}]
  \label{thm:main-monomer-dimer}
  Let $\mathcal{G}$ be a family of finite graphs of connective constant at
  most $\Delta$, and let $\gamma > 0$ be any fixed edge activity.
  Then there is an FPTAS for the partition function of the
  monomer-dimer model with edge activity $\gamma$ for all graphs in
  $\mathcal{G}$.  More specifically, the running time of the FPTAS for
  producing an $(1\pm \epsilon)$ factor approximation is
  $\inp{n/\epsilon}^{O( \sqrt{\gamma \Delta} \log \Delta)}$.
\end{theorem}
The previous best deterministic approximation algorithm for the
partition function of the monomer-dimer model was due to Bayati
\emph{et al.}~\cite{bayati_simple_2007}, and ran in time
$\inp{n/\epsilon}^{O( \sqrt{\gamma d} \log d)}$ for graphs of degree
at most $d+1$.  Thus, our algorithm replaces the maximum degree
constraint of Bayati \emph{et al.} by a corresponding constraint on
the connective constant, without requiring any bounds on the maximum
degree.  In particular, for graphs such as $\mathcal{G}(n, d/n)$
which have bounded connective constant and unbounded degree, our
analysis yields a polynomial time algorithm (for any fixed value of
the edge activity $\gamma$) in contrast to the sub-exponential time
algorithm obtained by Bayati \emph{et
  al.}~\cite{bayati_simple_2007}.  Using an observation of Kahn and
Kim~\cite{kahn_random_1998}, Bayati \emph{et al.} also pointed out
that the $\sqrt{d}$ factor in the exponent of their running time was
optimal for algorithms which are based on Weitz's framework and which use only
the fact that the maximum degree of the graph is at most $d+1$.  A
similar observation shows that the $\sqrt{\gamma\Delta}$ factor in
the exponent of our running time is optimal for algorithms in the
Weitz framework which use bounds on the connective constant (see the
remark at the end of Section~\ref{sec:spec-mess-monom} %
for a more detailed
discussion of this point).  As an aside, we also note that when no
bounds on the connective constant are available our FPTAS degrades
to a sub-exponential algorithm, as does the algorithm of Bayati
\emph{et al.} in the case of unbounded degree graphs.

\subsection{Techniques}
\label{sec:techniques}

The analyses by Weitz~\cite{Weitz06CountUptoThreshold} and Bayati
\emph{et al.}~\cite{bayati_simple_2007} both begin with the standard
observation that obtaining an FPTAS for the marginal probabilities of the
Gibbs distribution is sufficient in order to obtain an FPTAS
for the partition function.  The next non-trivial step is to show that
this computation of marginal probabilities at a given vertex $v$ in a
graph $G$ can be carried out on the tree $\Ts{v, G}$ of
\emph{self-avoiding walks} in $G$ starting at $v$.  Transferring the
problem to a tree allows one to write down a recurrence for the
marginal probabilities of a node in the tree in terms of the marginal
probabilities of its children.  However, since the tree is of
exponential size, one needs to truncate the tree at a small
(logarithmic) depth in order to obtain a polynomial time algorithm.
Such a truncation in turn introduces an ``error'' at the leaves.  The
challenge then is to show that this error contracts exponentially as
the recurrence works its way up to the root.

The approach of~\cite{Weitz06CountUptoThreshold,bayati_simple_2007}
(and similar results in~\cite{li_correlation_2011,
sinclair_approximation_2012,li_approximate_2012}) for establishing
this last condition takes the following general
form: one shows that at each step of the recurrence, the correlation
decay condition implies that the error at the parent node is less than
a constant factor (less than 1) times the maximum ($\ell_\infty$ norm)
of the errors at the children of the node.  Intuitively, this strategy
loses information about the structure of the tree by explaining the
error at the parent in terms of \emph{only one} of its children, and
hence it is not surprising that the results obtained from it are only
in terms of a local parameter such as the maximum degree.

The main technical contribution of \cite{sinclair13:_spatial} was to
show that, by analyzing the decay in terms of the $\ell_2$
norm---rather than the $\ell_\infty$ norm---of the errors at the
children, one can get past this limitation and obtain a result in terms of
the connective constant. Nevertheless, as stated above, the results
obtained in~\cite{sinclair13:_spatial} did not hold over the best
possible range of parameters.  Our main innovation in the present
paper is to analyze instead a norm \emph{adapted to} the parameters of
the model, rather than a fixed norm such as $\ell_2$ or $\ell_\infty$.
Specifically, we show that optimal results can be obtained by
analyzing the decay in terms of a carefully picked $\ell_q$ norm where $q$ is chosen
as a function of the connective constant and the model parameters (the
fugacity $\lambda$ in the hard core model and the edge activity
$\gamma$ in the monomer-dimer model).  At a technical level, the
use of these adaptive norms implies that we can no longer employ the
relatively simpler convexity arguments used
in~\cite{sinclair13:_spatial} in order to bound the propagation of
errors; characterizing the ``worst case'' error vectors now requires
solving a more involved optimization problem, which is the main new
technical challenge in this paper.
In Section~\ref{sec:decay-corr-saw}, we give a general framework for
tackling this problem.  Our model specific main results are then
obtained as direct applications of this framework.  We conjecture that
our framework may find applications to other approximate counting
problems as well.

\subsection{Related work}
\label{sec:related-work}

MCMC based algorithms for approximating the partition function of the
hard core model on graphs of bounded degree $d+1$ were obtained under
the condition $\lambda < 1/(d-2)$ by Luby and
Vigoda~\cite{luby_approximately_1997}, and later under the weaker
condition $\lambda < 2/(d-1)$ by Dyer and
Greenhill~\cite{dyer_markov_2000} and Vigoda~\cite{vigoda_note_2001}.
Weitz~\cite{Weitz06CountUptoThreshold} obtained an FPTAS under the
much weaker condition $\lambda < \lambda_c(d)$ by establishing a tight
connection between the algorithmic problem and the decay of
correlations on the $d$-ary tree.  This connection was further
tightened by Sly~\cite{Sly2010CompTransition} (see also Galanis
\emph{et al.}~\cite{Vigoda-hard-core-11} and Sly and Sun~\cite{sly12})
who showed that approximating the partition function of the hard core
model on $(d+1)$-regular graphs is NP-hard when $\lambda >
\lambda_c(d)$.  Weitz's algorithm was also one of the first
deterministic algorithms for approximating partition functions (along
with the contemporaneous work of Bandhopadhyay and
Gamarnik~\cite{BG08:_count_without_sampl}) which exploited decay of
correlations directly---in contrast to earlier algorithms which were
mostly based on MCMC techniques.  To date, no MCMC
based algorithms for the partition function of the hard core model
are known to have as large a range of applicability as Weitz's algorithm.

Weitz's approach has also been used to study the correlation decay
phenomenon on specific lattices.  Restrepo \emph{et
  al.}~\cite{restrepo11:_improv_mixin_condit_grid_count} supplemented
Weitz's approach with sophisticated computational arguments tailored
to the special case of $\mathbb{Z}^2$ to obtain strong spatial mixing
on $\mathbb{Z}^2$ under the condition $\lambda < 2.38$. Using even
more extensive numerical work, this bound was improved to $\lambda <
2.48$ by Vera \emph{et
  al.}~\cite{vera13:_improv_bound_phase_trans_hard}.  In contrast, a
direct application of Weitz's results yields the result only under the
condition $\lambda < \lambda_c(3) = 1.6875$.  Sampling from the hard
core model on special classes of unbounded degree graphs has also been
considered in the literature.  Mossel and Sly~\cite{mossel_gibbs_2010}
gave an MCMC based algorithm for sampling from the hard core model on
graphs drawn from $\mathcal{G}(n, d/n)$.  However, their algorithm is
only applicable in the range $\lambda <
O(1/d^2)$. Efthymiou~\cite{efthymiou13:_mcmc_g} gave an MCMC based
sampler under the much weaker condition $\lambda < 1/(2d)$.  %
{%
  Hayes and Vigoda~\cite{hayes_coupling_2006} also considered the
  question of sampling from the hard core model on special classes of
  unbounded degree graphs.  They showed that for regular graphs on $n$
  vertices of degree $d(n) = \Omega(\log n)$ and of girth greater than
  $6$, the Glauber dynamics for the hard core model mixes rapidly for
  $\lambda < e/d(n)$.  These results are incomparable to
  Theorem~\ref{thm:main-hard-core}.  The latter neither requires the
  graph to be regular, nor any lower bounds on
  its degree or girth, but it does require additional information about
  the graph in the form of its connective constant.  However, when the
  connective constant is available, then irrespective of the maximum
  degree of the graph or its girth, the theorem affords an FPTAS for
  the partition function.  %

  In contrast to the case of the hard core model, much more progress
  has been made on relating spatial mixing to notions of average
  degree in the case of the zero field \emph{ferromagnetic} Ising
  model.  Lyons~\cite{lyons_ising_1989} showed that on an arbitrary
  tree, a quantity similar in flavor to the \emph{connective
    constant}, known as the \emph{branching factor}, exactly
  determines the threshold for uniqueness of the Gibbs measure for
  this model. For the ferromagnetic Ising model on general graphs,
  Mossel and Sly~\cite{mossel_rapid_2009,mossel_exact_2013} proved
  results analogous to our Theorem \ref{thm:main-hard-core}. An
  important ingredient in the arguments in
  both~\cite{lyons_ising_1989} and
  ~\cite{mossel_rapid_2009,mossel_exact_2013} relating correlation
  decay in the zero field Ising model to the branching factor and the
  connective constant is the symmetry of the ``$+$'' and ``$-$'' spins
  in the zero field case. In work related to \cite{lyons_ising_1989},
  Pemantle and Steif~\cite{pemantle_robust_1999} defined the notion of
  a \emph{robust phase transition (RPT)} and related the threshold for
  RPT for various ``symmetric'' models such as the zero field Potts
  model and the Heisenberg model on general trees to the branching
  factor of the tree.  Again, an important ingredient in their
  arguments seems to be the existence of a symmetry on the set of
  spins under whose action the underlying measure remains invariant.
  In contrast, in the hard core model, the two possible spin states of
  a vertex (``occupied'' and ``unoccupied'') do not admit any such
  symmetry. %
} %

{%
A preliminary version of this paper~\cite{sinclair13:_spatial}
investigated decay %
of correlations in the hard core model in graphs with possibly
unbounded degree but bounded connective constant.  There, it was shown
that when $\lambda < \lambda_c(\Delta+1)$, all graphs of connective
constant at most $\Delta$ exhibit decay of correlations and also admit
an FPTAS for the partition function.  The observation that for any
$\epsilon > 0$, the connective constant of graphs drawn from
$\mathcal{G}(n, d/n)$ is at most $d(1+\epsilon)$ with high probability
was then used to derive a polynomial time sampler for the hard core
model on $\mathcal{G}(n, d/n)$ all the way up to $\lambda < e/d$
(which was the conjectured asymptotic bound).
Even for the case of bounded degree graphs, it was shown
in~\cite{sinclair13:_spatial} that estimates
on the connective constant could be used to improve the range of
applicability of Weitz's results.  This latter result was then used in
conjunction with well known upper bounds on the connective constants of various
regular lattices to improve upon the known
strong spatial mixing bounds for those lattices (these included $\Z^d$
for $d \geq 3)$.  However, in the special case of $\Z^2$, the bounds
in~\cite{sinclair13:_spatial} fell short of those obtained by Restrepo \emph{et
  al.}~\cite{restrepo11:_improv_mixin_condit_grid_count} and Vera
\emph{et al.}~\cite{vera13:_improv_bound_phase_trans_hard}.

The results of the present paper for the hard core model strengthen
and unify the two distinct results of~\cite{sinclair13:_spatial}
mentioned above, by replacing the
maximum degree constraint in Weitz's result completely by an exactly
corresponding constraint on the connective constant: in particular, we
show that graphs of connective constant $\Delta$ admit an FPTAS (and
exhibit strong spatial mixing) whenever $\lambda <
\lambda_c(\Delta)$. Thus, for example, our results extend the range of
applicability of the above quoted results of~\cite{sinclair13:_spatial}
for $\mathcal{G}(n, d/n)$ to $\lambda <
\lambda_c(d)$; since $\lambda_c(d) > e/d$ for all $d \geq 1$, this is
a strict improvement on~\cite{sinclair13:_spatial}.
Regarding the question of improved strong spatial mixing bounds on
specific lattices, our results show that connective constant
computations are sufficient to improve upon the results of Restrepo
\emph{et al.}~\cite{restrepo11:_improv_mixin_condit_grid_count} and
Vera \emph{et al.}~\cite{vera13:_improv_bound_phase_trans_hard} even
in the case of $\mathbb{Z}^2$, without taking recourse to any further
numerical or computational work.
}%

In contrast to the case of the hard core model, where algorithms with
the largest range of applicability are already deterministic, much
less is known about the deterministic approximation of the partition
function of the monomer-dimer model.  Jerrum and
Sinclair~\cite{jerrum_approximating_1989} gave an MCMC-based FPRAS for
the monomer-dimer model which runs in polynomial time on all graphs
(without any bounds on the maximum degree) for any fixed value of the
edge activity $\lambda$.  However, no deterministic algorithms with
this range of applicability are known, even in the case of specific
graph families such as $\mathcal{G}(n, d/n)$.  So far, the best result
in this direction is due to Bayati, Gamarnik, Katz, Nair and
Tetali~\cite{bayati_simple_2007}, who gave an algorithm which produces
a $(1\pm \epsilon)$ factor approximation for the monomer-dimer
partition function in time $\inp{{n}/{\epsilon}}^{\tilde{O}(\sqrt{\gamma d})}$
on graphs of maximum degree $d$.  Their result therefore yields
super-polynomial (though sub-exponential) algorithms in the case of
graphs such as those drawn from $\mathcal{G}(n, d/n)$, which have
unbounded maximum degree even for constant $d$.  In contrast, the
present paper shows that the same running time can in fact be obtained
for graphs of \emph{connective constant} $d$, irrespective of the
maximum degree. The running times obtained by applying the results of
this paper to the case of $\mathcal{G}(n, d/n)$ are therefore
polynomial (and not merely sub-exponential) when $d$ is a fixed
constant.

The first reference to the connective constant occurs in classical
papers by Hammersley and Morton~\cite{hammersley_poor_1954},
Hammersley and Broadbent \cite{broadbent_percolation_1957} and
Hammersley~\cite{hammersley_percolation_1957}.  Since then, several
natural combinatorial questions concerning the number and other
properties of self-avoiding walks in various lattices have been
studied in depth; see the monograph of Madras and
Slade~\cite{madras96:_self_avoid_walk} for a survey.  %
Much work has been devoted especially to finding rigorous upper and
lower bounds for the connective constant of various
lattices~\cite{alm_upper_2005, alm_upper_1993,
  jensen_enumeration_2004, kesten_number_1964, poenitz00:_improv_z}.
Heuristic techniques from physics have also been brought to bear upon
this question. For example, Nienhuis~\cite{nienhuis_exact_1982}
conjectured on the basis of heuristic arguments that the connective
constant of the honeycomb lattice $\mathbb{H}$ must be
$\sqrt{2+\sqrt{2}}$.  Nienhuis's conjecture was rigorously confirmed
in a recent paper of Duminil-Copin and
Smirnov~\cite{duminil-copin_connective_2012}.

\section{Preliminaries}
\label{sec:preliminaries}

\subsection{Probabilities and likelihood ratios}
\label{sec:prob-likel-rati}
In this section, we define some standard marginals of the
monomer-dimer and hard core distributions.  The importance of these
quantities for our work comes from the standard ``self-reducibility''
arguments which show that obtaining an FPTAS for these marginals is
sufficient in order to obtain an FPTAS for the partition function of
these models (see Section~\ref{sec:from-prob-part} for these
reductions).

\subsubsection{Monomer-dimer model}
\label{sec:monomer-dimer-model}
\begin{definition}[\textbf{Monomer probability}]
  Consider the Gibbs distribution of
  the monomer-dimer model with dimer activity $\gamma$ on a finite
  graph $G = (V,E)$, and let $v$ be a vertex in $V$.  We define the
  \emph{monomer probability} $p(v, G)$ as
  \begin{displaymath}
    p_v \defeq \Pr{v \not\in M},
  \end{displaymath}
  which is the probability that $v$ is unmatched (i.e., a
  \emph{monomer}) in a matching $M$ sampled from the Gibbs
  distribution.
\end{definition}
\begin{remark}
  The monomer-dimer model is often described in the literature in
  terms of a \emph{monomer activity} $\lambda$ instead of the
  \emph{dimer activity} $\gamma$ used here.  In this formulation, the
  weight of a matching $M$ is $\lambda^{u(M)}$, where $u(M)$ is the
  number of unmatched vertices (\emph{monomers}) in $M$.  The two
  formulations are equivalent: with dimer activity $\gamma$
  corresponding to monomer activity $\lambda = \frac{1}{\gamma^2}$.
\end{remark}

\subsubsection{Hard core model}
\label{sec:hard-core-model}

In the case of the hard core model, we will need to define the
appropriate marginals in a slightly more generalized setting.  Given a
graph $G = (V,E)$, a \emph{boundary condition} will refer to a
partially specified independent set in $G$; formally, a boundary
condition $\sigma = (S, I)$ is a subset $S \subseteq V$ along with an
independent set $I$ on $S$. (Boundary conditions may be seen as
special instances of initial conditions defined above.)

\begin{definition}[\textbf{Occupation probability and occupation
    ratio}]
  Consider the
  hard core model with vertex activity $\lambda > 0$ on a finite graph $G$,
  and let $v$ be a vertex in $G$.  Given a boundary condition $\sigma = (S,
  I_S)$ on $G$, the \emph{occupation probability} $p_v(\sigma, G)$ at
  the vertex $v$ is the probability that $v$ is included in an independent
  set $I$ sampled according to the hard core distribution conditioned
  on the event that $I$ restricted to $S$ coincides with $I_S$.  The
  \emph{occupation ratio} $R_v(\sigma,G)$ is then defined as
  \begin{displaymath}
    R_v(\sigma,G) = \frac{p_v(\sigma,G)}{1-p_v(\sigma,G)}.
  \end{displaymath}
\end{definition}

\subsubsection{Strong Spatial Mixing}
\label{sec:strong-spat-mixing}

We present the definition of strong spatial mixing for the special
case of the hard core model; the definition in the case of the monomer
dimer model is exactly analogous.  Our definition here closely follows
the version used by Weitz~\cite{Weitz06CountUptoThreshold}.
\begin{definition}
  \textbf{(Strong Spatial Mixing).}  The hard core model with a fixed
  vertex activity $\lambda > 0$ is said to exhibit \emph{strong
    spatial mixing} on a family $\family{F}$ of graphs if for any
  graph $G$ in $\family{F}$, any vertex $v$ in $G$, and any two
  boundary conditions $\sigma$ and $\tau$ on $G$ which differ only at
  a distance of at least $\ell$ from $v$, we have
  \begin{equation*}
    \abs{R_v(\sigma, G) - R_v(\tau, G)} = O(c^\ell).
  \end{equation*}
  for some fixed constant $0 \leq c < 1$.
\end{definition}
An important special condition of the definition is when the family
$\mathcal{F}$ consists of a single infinite graph (e.g.,
$\mathbb{Z}^2$ or another regular lattice).  The constant $c$ in the
definition is often referred to as the \emph{rate} of strong spatial
mixing.

\subsection{Truncated recurrences with initial conditions}
\label{sec:trunc-recurr-with}

As in the case of other correlation decay based algorithms~(e.g., in
\cite{gamarnik_correlation_2007, Weitz06CountUptoThreshold}), we will
need to analyze recurrences for marginals on rooted trees with various initial conditions.
We therefore set up some notation for describing such
recurrences.  For a vertex $v$ in a tree $T$, we will denote
by $|v|$ the distance of $v$ from the root of the tree.
Similarly, for a set $S$ of vertices, $\delta_S \defeq \min_{v \in S}
|v|$.
\begin{definition}[\textbf{Cutset}]
  Let $T$ be any tree rooted at $\rho$. A
  \emph{cutset} $C$ is a set of vertices in $T$ satisfying the
  following two conditions: %
  \begin{enumerate}
  \item Any path from $\rho$ to a leaf $v$ with $|v| \geq \delta_C$
    must pass through $C$.
  \item The vertices in $C$ form an antichain, i.e., for any vertices
    $u$ and $v$ in $C$, neither vertex is an ancestor of the other in
    $T$.
  \end{enumerate}

  A trivial example of a cutset is the set $L$ of all the leaves of
  $T$.  Another example we will often need is the set $S_\ell$ of all
  vertices at distance $\ell$ from $\rho$ in $T$.
\end{definition}
\begin{remark}
  In Section~\ref{sec:more-results-trees}, we will need to use cutsets
  on rooted locally finite infinite trees.  In this case, the first
  condition in the definition changes to ``any infinite path starting
  from the root $\rho$ must pass through $C$.''
\end{remark}
{ %
For a cutset $C$, we denote by $T_{\leq C}$ the subtree of $T$
obtained by removing the descendants of vertices in $C$ from $T$, and
by $T_{< C}$ the subtree of $T$ obtained by removing the vertices in
$C$ from $T_{\leq C}$.  Further, for a vertex $u$ in $T$, we denote by
$T_u$ the subtree of $T$ rooted at $u$, and by $T_{u, \leq C}$ and
$T_{u, < C}$ the intersections of $T_u$ with $T_{\leq C}$ and $T_{<
  C}$ respectively. %
}%
\begin{definition}[\textbf{Initial condition}]
  An \emph{initial condition} $\sigma = (S, P)$ is a set $S$ of
  vertices in $T$ along with an assignment $P: S \rightarrow [0,b]$ of
  bounded positive values to vertices in $S$.
\end{definition}
We are now ready to describe the tree recurrences.  Given an initial
condition $\sigma = (S, P)$ along with a default value $b_0$ for the
leaves, a family of functions $f_d:[0,b]^d
\rightarrow [0,b]$ for every positive integer $d \geq 1$, and a vertex
$u$ in $T$, we let $F_{u}(\sigma)$ denote the value obtained at $u$ by
iterating the tree recurrences $f$ on the subtree $T_{u}$ rooted at $u$ under the
initial condition $\sigma$.  Formally, we define $F_u(\sigma) = b_0$
when $u\not \in S$ is a leaf, and
\begin{equation}
  F_u(\sigma) =
  \begin{cases}
    P(u) & \text{when $u \in S$,}\\
    f_{d}\inp{F_{u_1}(\sigma),\dotsc,F_{u_d}(\sigma)}&\parbox[c]{0.3\textwidth}{when
      $u \not\in S$ is of arity $d\geq 1$ and has children $u_1, u_2,
      \dotsc, u_d$.}\label{eq:5}
  \end{cases}
\end{equation}
\subsection{The self-avoiding walk tree and associated recurrences}
\label{sec:tree-recurrence}

Given a vertex $v$ in a graph $G$, one can define a rooted tree
$\Ts{v,G}$ of self-avoiding walks (called the \emph{self-avoiding walk
  tree}, or \emph{SAW tree}) starting at $v$, as follows: the root of the
tree represents the trivial self-avoiding walk that ends at $v$, and
given any node $u$ in the tree, its children represent all possible
self-avoiding walks than can be obtained by extending the
self-avoiding walk represented by $u$ by exactly one step.  The
importance of the self-avoiding walk tree for computation stems from
the beautiful results of Godsil~\cite{godsil_matchings_1981} (for the
monomer-dimer model) and Weitz~\cite{Weitz06CountUptoThreshold} (for
the hard core model), which allow the derivation of simple recurrences
for the monomer probability $p_v(G)$ on general graphs.  We begin
with the case of the monomer-dimer model.
\subsubsection{Monomer-dimer model}
\begin{theorem}[\textbf{Godsil~\cite{godsil_matchings_1981}}]
\label{thm:saw-tree}
Let $v$ be a vertex in a graph $G$, and consider the monomer-dimer
model with dimer activity $\gamma > 0$ on the graphs $G$ and $\Ts{v,
  G}$. We then have
\begin{equation*}
p_v(G) = p_v(\Ts{v, G}).
\end{equation*}
\end{theorem}
The promised recurrence for $p_v(G)$ can now be derived using dynamic
programming on the tree $\Ts{v, G}$.  In particular, let $T$ be any
tree rooted at $\rho$, and let $\rho_i$, $1 \leq i \leq d$ be the
children of $\rho$.  Denoting by $p_i$ the monomer probability
$p_{\rho_i}(T_{\rho_i})$ at the root of the subtree $T_{\rho_i}$, one can then show that~(see, e.g.,
\cite{kahn_random_1998})
\begin{equation}
  p_\rho(T) = f_{d,\gamma}(p_1, p_2, \ldots, p_d) \defeq
  \frac{1}{1 + \gamma\sum_{i=1}^dp_i}.
\end{equation}
\begin{remark}
  In what follows, we will often suppress the dependence of
  $f_{d,\gamma}$ on $\gamma$ for convenience of notation.
\end{remark}
In terms of our notation for tree recurrences, we note that the actual
computation of $p_\rho(T)$ corresponds to computing
$F_\rho(\vec{1}_L)$, where the initial condition $\vec{1}_L$ assigns
the value $1$ to all vertices in $L$, the cutset comprising all the
leaves (and with the boundary value $b_0$ set to $1$), since the base
case of the recurrence comprises a single vertex which has monomer
probability $1$ by definition.

Note that the self-avoiding tree can be of exponential size, so that
Godsil's reduction does not immediately yield an efficient algorithm
for computing $p_\rho(G)$.  In order to obtain an algorithm, we would
need to consider truncated versions of the recurrence, obtained by
specifying initial conditions on the cutset $S_\ell$ comprising all
vertices at distance $\ell$ from $\rho$.  Since $f_{d,
  \gamma}$ is monotonically decreasing in each of its arguments, we have
\begin{equation}
  \begin{aligned}
    F_\rho(\vec{0}_\ell) \leq p_\rho(T) \leq F_\rho(\vec{1}_\ell)
    &\quad\text{when $\ell$ is even, and}\\
    F_\rho(\vec{0}_\ell) \geq p_\rho(T) \geq F_\rho(\vec{1}_\ell)
    &\quad\text{when $\ell$ is odd.}\\
  \end{aligned}
\end{equation}
Here, the initial condition $\vec{0}_\ell$ (respectively,
$\vec{1}_\ell$) assigns the value $0$ (respectively, 1) to every
vertex in $S_\ell$.  Given these conditions, it is sufficient to show
that the difference between $F_\rho(\vec{0}_\ell)$ and
$F_\rho(\vec{1}_\ell)$ decreases exponentially in $\ell$ in order to
establish that truncated versions of the recurrence converge to the
true answer $p_\rho(T)$ exponentially fast in the ``truncation
length'' $\ell$.

\subsubsection{Hard core model}
\label{sec:hard-core-model-1}
Weitz~\cite{Weitz06CountUptoThreshold} proved a reduction similar to
that of Godsil for the hard core model. However, in contrast to
Godsil's reduction for the monomer-dimer model Weitz's reduction
requires a boundary condition to be applied to the self avoiding walk
tree.
\begin{theorem}[\textbf{Weitz~\cite{Weitz06CountUptoThreshold}}]
\label{thm:weitz-saw-tree}
Let $v$ be a vertex in a graph $G$, and consider the hard core model
with vertex activity $\lambda > 0$ on the graphs $G$ and $\Ts{v,
  G}$. Then, there exists an efficiently computable boundary condition
$\mathcal{W}$ on $\Ts{v, G}$ such that for any boundary condition
$\sigma$ on $G$, we have
\begin{equation}
  R_v(\sigma, G) = R_v(\mathcal{W} \cup \sigma, \Ts{v, G}),\label{eq:2}
\end{equation}
where (1) the boundary condition $\sigma$ on the right hand side
denotes the natural translation of the boundary condition $\sigma$ on
$G$ to $\Ts{v, G}$, and (2) $\mathcal{W} \cup \sigma$ is the boundary
condition obtained by first applying the boundary condition
$\mathcal{W}$, and then $\sigma$ (that is, $\sigma$ overrides
$\mathcal{W}$ on vertices on which $\sigma$ and $\mathcal{W}$ are both
specified and disagree).
\end{theorem}
\noindent We will often refer to a self-avoiding walk tree with
Weitz's boundary condition as ``a Weitz SAW tree''.

As in the case of the monomer-dimer model, the theorem allows the
computation of $R_v(\sigma, G)$ using natural recurrences on the tree.
Using the same notation as in the case of the monomer-dimer model, we
denote $R_{\rho_i}(\sigma, T_{\rho_i})$ as $R_i$.  It is well known
that (see, e.g., \cite{Weitz06CountUptoThreshold}) that
\begin{equation}
  R_{\rho}(\sigma, T) = f_{d,\lambda}(R_1, R_2, \ldots, R_d) \defeq
  \lambda\prod_{i=1}^d\frac{1}{1+R_i}.\label{eq:14}
\end{equation}
We now see that in terms of our
notation for tree recurrences, the computation of $R_\rho(\sigma, T)$
corresponds to computing $F_\rho(\vec{\lambda}_L \cup \sigma)$ (with
the boundary value $b_0$ for leaves set to $\lambda$), where the
initial condition $\vec{\lambda}_L \cup \sigma$ assigns the value
$\lambda$ to all vertices in the set $L$ of leaves, and then applies
the boundary condition $\sigma$ (possibly overriding previously
assigned values).  Note that the boundary condition $\sigma$ assigns
$R_v = \infty$ for vertices $v$ which are set to occupied by $\sigma$,
and hence, strictly speaking, violates the requirement that initial
conditions should only assign bounded values.  However, this can be
fixed easily by observing that an initial condition which assigns $R_v
= \infty$ is equivalent to the one which assigns $R_u = 0$ to the
parent $u$ of $v$.  Thus, we may assume without loss of generality
that our initial conditions only assign values from the interval $[0,
\lambda]$.

Again, as in the case of the monomer-dimer model, we will need to work
with truncated trees.  As before, we consider initial condition
specified on cutsets $S_\ell$ of vertices at distance $\ell$ from the
root $\rho$, and use the fact that $f_{d, \lambda}$ is monotonically
decreasing in each of its arguments to see that
\begin{equation}
  \begin{aligned}
    F_\rho(\vec{0}_\ell\cup \sigma) \leq R_\rho(\sigma, T) \leq
    F_\rho(\vec{\lambda}_\ell\cup \sigma)
    &\quad\text{when $\ell$ is even, and}\\
    F_\rho(\vec{0}_\ell\cup \sigma) \geq R_\rho(\sigma, T) \geq
    F_\rho(\vec{\lambda}_\ell\cup \sigma)
    &\quad\text{when $\ell$ is odd.}\\
  \end{aligned}
\end{equation}
Here, the initial condition $\vec{0}_\ell \cup \sigma$ (respectively,
$\vec{\lambda}_\ell \cup \sigma$) assigns the value $0$ (respectively,
$\lambda$) to every vertex in $S_\ell$, after which the boundary
condition $\sigma$ is applied, possibly overriding the earlier
assignments (from the previous discussion, we can assume that the
effect of $\sigma$ is limited o setting some more vertices to $0$).
As before, it is then sufficient to show that the difference between
$F_\rho(\vec{0}_\ell \cup \sigma)$ and $F_\rho(\vec{1}_\ell \cup
\sigma)$ decreases exponentially in $\ell$ in order to establish that
truncated versions of the recurrence converge to the true answer
$R_\rho(\sigma, T)$ exponentially fast in the ``truncation length''
$\ell$.

%
\subsection{From probabilities to the partition function}
\label{sec:from-prob-part}
In this section, we review some standard facts on how approximation
algorithms for the marginal probabilities translate into approximation
algorithms for the partition function~(see, e.g,
\cite{Weitz06CountUptoThreshold, gamarnik_correlation_2007}).  We
provide the calculations here for the case of the monomer-dimer model,
and refer to Weitz~\cite{Weitz06CountUptoThreshold} for similar
calculations for the hard core model.

Let $v_1, v_2, \dotsc, v_n$ be any arbitrary ordering of the vertices
of $G$.  Since the monomer-dimer partition function of the empty graph
is $1$, we then have
\begin{align}
  Z(G) &= \prod_{i=1}^n \frac{
    Z\inp{
      G - \inb{v_1, \dotsc, v_{i-1}}
    }
  }{
    Z\inp{
      G - \inb{v_1, \dotsc, v_{i}}
    }
  }\nonumber\\
  &=\prod_{i=1}^n \frac{
    1
  }{
    p_{v_i}\inp{
      G - \inb{v_1,\dotsc,v_{i-1}}
    }
  }.
\end{align}
Suppose, we have an FPTAS for the probabilities $p_\rho$ which runs in
time $t(n, 1/\epsilon)$ and produces an output $\hat{p}$ such that
$p_\rho/(1+\epsilon) \leq \hat{p} \leq p_\rho$. Now, given $\epsilon
\leq 1$, we use the FPTAS in time $t\inp{n, {2n}/{\epsilon}}$ to
compute an approximation $\hat{p}_i$ to the $p_{v_i}\inp{G -
  \inb{v_1,\dotsc,v_{i-1}}}$.  We then have for each $i$
\begin{displaymath}
  \frac{1}{p_{v_i}\inp{G - \inb{v_1,\dotsc,v_{i-1}}}}
  \leq \frac{1}{\hat{p}_i}
  \leq  \frac{1+ \epsilon/(2n)}{p_{v_i}\inp{G - \inb{v_1,\dotsc,v_{i-1}}}}.
\end{displaymath}
By multiplying these estimates. we obtain an estimate $\hat{Z}$ of the
partition function which satisfies
\begin{equation*}
  Z(G) \leq \hat{Z} \leq Z(G)\inp{1 + \frac{\epsilon}{2n}}^n \leq
  Z(G)e^{\epsilon/2} \leq Z(G)(1+\epsilon),
\end{equation*}
where we use the condition $\epsilon \leq 1$ in the last inequality.
Thus, the total running time is $O\inp{n\cdot t\inp{n,
    {2n}/{\epsilon}}}$, which is polynomial in $n$ and
$1/\epsilon$ whenever $t$ is.  Thus, it is sufficient to derive an
FPTAS for the marginal probabilities in order to obtain an FPTAS for
the partition function.

\subsection{The connective constant}
\label{sec:connective-constant}
We now recall the definition of the connective constant of a
graph. Given a vertex $v$ in a locally finite graph, we denote by
$N(v, l)$ the number of self-avoiding walks of length $l$ in the graph
which start at $v$.  The connective constant for infinite graphs is
then defined as follows.
\begin{definition}[\textbf{Connective constant: infinite
    graphs~\cite{madras96:_self_avoid_walk}}] Let $G=(V,E)$ be a
  locally finite infinite graph.  The \emph{connective constant}
  $\Delta(G)$ of $G$ is $\sup_{v\in
    V}\limsup_{\ell\rightarrow\infty}N(v,\ell)^{1/\ell}$.
\end{definition}
\begin{remark}
  The supremum over $v$ in the definition is clearly not required for
  vertex-transitive graphs such as Cartesian lattices. Further, in
  such graphs the $\limsup$ can be replaced by a
  limit~\cite{madras96:_self_avoid_walk}.
\end{remark}
The definition was extended in~\cite{sinclair13:_spatial} to families
of finite graphs parametrized by size. As observed there, such a
parametrization is natural for algorithmic applications.
\begin{definition}[\textbf{Connective constant: finite
    graphs~\cite{sinclair13:_spatial}}] Let $\family{F}$ be a family
  of finite graphs.  The connective constant of $\family{F}$ is at
  most $\Delta$ if there exist constants $a$ and $c$ such that for any
  graph $G = (V,E)$ in $\family{F}$ and any vertex $v$ in $G$, we have
  $\sum_{i=1}^\ell N(v, i) \leq c\Delta^\ell$ for all $\ell \geq a\log
  |V|$.
\end{definition}
It is easy to see that the connective constant of a graph of maximum
degree $d+1$ is at most $d$.  However, the connective constant can be
much smaller than the maximum degree. For example, though the maximum
degree of a graph drawn from the Erdős–Rényi model $\G\inp{n, d/n}$ is
$\Theta(\log n/\log\log n)$ w.h.p, it is not hard to show
(see~\cite{sinclair13:_spatial}) that for any fixed $\epsilon > 0$,
the connective constant of such a graph is at most $d(1+\epsilon)$
w.h.p.

\begin{remark}
  Note that the connective constant has a natural interpretation as
  the ``average arity'' of the SAW tree, since vertices in $\Ts{v, G}$
  at distance $\ell$ from the root are in bijection with self-avoiding
  walks of length $\ell$ starting at $v$.
\end{remark}

\section{Decay of correlations on the SAW tree}
\label{sec:decay-corr-saw}

In this section, we lay the groundwork for proving decay of
correlations results for the tree recurrences $F_{\rho}$
defined in eq.~(\ref{eq:5}) for both the hard core and monomer-dimer
models: such a result basically affirms that truncating the recurrence
at a small depth $\ell$ is sufficient in order to approximate
$F_\rho$ with good accuracy.  Our proof will use the
\emph{message
  approach}~\cite{restrepo11:_improv_mixin_condit_grid_count,
  li_correlation_2011, sinclair_approximation_2012}, which proceeds by
defining an appropriate function $\phi$ of the marginals being
computed and then showing a decay of correlation result for this
function.
\begin{definition}[\textbf{Message~\cite{restrepo11:_improv_mixin_condit_grid_count, sinclair_approximation_2012, li_correlation_2011}}] Given a
  positive real number $b$, a \emph{message} is a strictly increasing
  and continuously differentiable function $\phi: (0, b] \rightarrow
  \mathbb{R}$, with the property that the derivative of $\phi$ is
  bounded away from $0$ on its domain.  A message $\phi$ is guaranteed
  to admit a continuously differentiable inverse, which we will denote
  by $\psi$.
\end{definition}
In the rest of this section, we will work in the abstract
framework described in Section~\ref{sec:trunc-recurr-with}, to
illustrate how the message approach can be used to get strengthened decay of
correlation estimates as compared to those obtained from direct
analyses of one step of the recurrence. We will then instantiate our
framework with appropriately chosen messages for the monomer-dimer and
the hard core models in Sections~\ref{sec:spec-mess-hard}
and~\ref{sec:spec-mess-monom}.

We begin by fixing the boundary value $b_0$ for the leaves in our
recurrence framework, and assume that the initial conditions specify
values in the interval $[0, b]$.  We assume that we have a set of tree
recurrences $f_d:[0,b]^d \rightarrow [0,b]$ for every positive integer
$d \geq 1$.  The only constraints we put on the recurrences in this
section are the following (both of which are trivially satisfied by
the recurrences for the hard core and the monomer-dimer model).
\begin{condition}[\textbf{Consistency}]
  We say that a set of recurrences $\inb{f_d}_{d \geq 1}$, where $f_d$
  is $d$-variate, are \emph{consistent} if they obey the following two
  conditions:
  \begin{enumerate}
  \item If $\vec{x} \in \R^d$ is a permutation of $\vec{y}\in\R^d$,
    then $f_d(\vec{x}) = f_d(\vec{y})$.
  \item If all but the first $k$ co-ordinates of $\vec{x} \in \R^d$
    are $0$, then $f_d(\vec{x}) = f_k(x_1, x_2, x_3, \ldots, x_k)$.
  \end{enumerate}
\end{condition}
\noindent Given the message $\phi$ (and its inverse $\psi$), we
further define $f_d^\phi$ by
\begin{displaymath}
  f_d^\phi(x_1,x_2,\ldots,x_d) \defeq \phi\inp{f_d\inp{\psi(x_1),
      \psi(x_2),\ldots,\psi(x_d)}}.
\end{displaymath}
We then have the following simple consequence of the mean value
theorem~(a proof can be found in Appendix~\ref{sec:proof-lemma-refl}).
\begin{lemma}[\textbf{Mean value theorem}]
  \label{lem:mean-value}
  Consider two vectors $\vec{x}$ and $\vec{y}$ in $\phi([0,B])^d$.
  Then there exists a vector $\vec{z} \in [0, \infty)^d$ such that
  \begin{equation}
    \abs{f_{d}^\phi(\vec{x}) - f_{d}^\phi(\vec{y})} \leq
    \Phi\inp{f_{d}(\vec{z})}
    \sum_{i=1}^d
    \frac{\abs{y_i-x_i}}{\Phi(z_i)} \abs{\pdiff{f_{d}}{z_i}},
    \label{eq:10}
  \end{equation}
  where $\Phi := \phi'$ is the derivative of $\phi$, and by a slight
  abuse of notation we denote by $\pdiff{f_{d}}{z_i}$ the
  partial derivative of $f_{d}(R_1, R_2,\ldots,R_d)$ with
  respect to $R_i$ evaluated at $\vec{R} = \vec{z}$.
\end{lemma}

The first step of our approach is similar to that taken in the
papers~\cite{restrepo11:_improv_mixin_condit_grid_count,
  sinclair_approximation_2012,
  li_approximate_2012,li_correlation_2011} in that we will use an
appropriate message---along with the estimate in
Lemma~\ref{lem:mean-value}---to argue that the ``distance'' between
two input message vectors $\vec{x}$ and $\vec{y}$ at the children of a
vertex shrinks by a constant factor at each step of the recurrence.
Previous works~\cite{restrepo11:_improv_mixin_condit_grid_count,
  sinclair_approximation_2012, li_approximate_2012,
  li_correlation_2011} showed such a decay on some version of the
$\ell_\infty$ norm of the ``error'' vector $\vec{x} - \vec{y}$: this
was achieved by bounding the appropriate dual $\ell_1$ norm of the
gradient of the recurrence.  Our intuition is that in order to achieve
a bound in terms of a global quantity such as the connective constant,
it should be advantageous to use a more global measure of the error
such as an $\ell_q$ norm for some $q < \infty$.

{%
  In line with the above plan, we will attempt to bound the right hand
  side of eq.~(\ref{eq:10}) in terms of $\norm[q]{\vec{x}-\vec{y}}$
  for an appropriate value of $q < \infty$ by maximizing the sum while
  keeping $f_{d}(\vec{z})$ fixed. In the special case $q=2$, it is in
  fact posible to carry out this maximization using relatively simple
  concavity arguments.  This was the approach taken
  in~\cite{sinclair13:_spatial}, but the restriction $q=2$ leads to
  sub-optimal results.  Here we get past this limitation by using a
  more flexible optimization than that used
  in~\cite{sinclair13:_spatial}.  To do this, we will seek to
  establish the following property for our messages (the exponent $a$
  will be the H\"{o}lder conjugate of the value of $q$ that we
  eventually use). %
}
\begin{definition}
  Given a consistent family of recurrences $\inb{f_d}_{d \geq 1}$, a
  message $\phi$ (with $\Phi \defeq \phi'$) is said to be
  \emph{symmetrizable with exponent $a$} with respect to the family if
  it satisfies the following two conditions:
\begin{enumerate}
\item Let $\mathcal{D}$ be the domain of the recurrence family. For
  every positive integer $d$ and every real $B > 0$ for which the
  program
  \begin{align}
    \max\qquad &\sum_{i=1}^d
    \inp{
      \frac{1}{\Phi(x_i)}
      \abs{\pdiff{f_{d}}{x_i}}
    }^a,
    \qquad \text{where}\nonumber\\
    &f_d(\vec{x}) = B\nonumber\\
    &x_i \in \mathcal{D},\qquad 1\leq i\leq d\nonumber
  \end{align}
  is feasible, it also has a solution $\vec{x}$ in which all the
  non-zero entries of $\vec{x}$ are equal. (We assume implicitly that
  $0 \in \mathcal{D}$.)
\item
  $\lim_{x_i\rightarrow{0^+}}\frac{1}{\Phi(x_i)}\abs{\pdiff{f_d}{x_i}}
  = 0$ for all $d \geq 1$, and for any fixed values of the $x_j$, $j
  \neq i$.
\end{enumerate}

\end{definition}
For symmetrizable messages, we will be able to bound the quantity
$\vert f_{d}^\phi(\vec{x}) - f_{d}^\phi(\vec{y})\vert$ in terms of
$\norm[q]{\vec{x} - \vec{y}}$, where $1/a + 1/q = 1$, and our improved
correlation decay bounds will be based on the fact that
symmetrizability can be shown to hold under a wider range of values of
$q$ than that required by the concavity conditions used
in~\cite{sinclair13:_spatial}. Our bounds will be stated in terms of
the following notion of decay.

\paragraph{\bf Notation.} Given a $d$-variate function $f_d$ and a scalar
$x$, we denote by $f_d(x)$ the quantity $f_d(x, x, \ldots, x)$.

\begin{definition}[\textbf{Decay factor $\alpha$}]
  Let $\phi$ be a message with
  derivative $\Phi$, and let $a$ and $q$ be positive reals such that
  $\frac{1}{a} + \frac{1}{q} = 1$.  We define the functions
  $\Xi_{\phi,q}(d, x)$ and $\xi_{\phi,q}(d)$ as follows:
  \begin{align*}
    \Xi_{\phi,q}(d, x) &\defeq
    \frac{1}{d}\inp{\frac{\Phi(f_d(x))\abs{f_d'(x)}}{\Phi(x)}}^q;\\
    \xi_{\phi,q}(d) &\defeq \sup_{x \geq 0}\Xi_{\phi,q}(d, x).
  \end{align*}%
  The \emph{decay factor} $\alpha$ is then defined as
  \begin{equation}
    \alpha \defeq \sup_{d\geq 1}\xi_{\phi,q}(d).\label{eq:8}
  \end{equation}
\end{definition}
Armed with the above definitions, we are now ready to prove
Lemma~\ref{lem:tech}, which provides the requisite decay bound for one
step of the tree recurrence. The main technical step in applying this
lemma is to find $a, q$ as in the definition and a message $\phi$
symmetrizable with exponent $a$ for which the decay factor $\alpha$ is
small; Lemma~\ref{lem:general-tree} below then shows how the decay
factor comes into play in proving exponential decay of correlations
over the tree.
\begin{lemma}
  \label{lem:tech}
  Let $\phi$ be a message with derivative $\Phi$, and let $a$ and $q$
  be positive reals such that $\frac{1}{a} + \frac{1}{q} = 1$. If
  $\phi$ is symmetrizable with exponent $a$, then for any two vectors
  $\vec{x}, \vec{y}$ in $\phi([0, b])^d$, there exists an integer $k
  \leq d$ such that
  \begin{displaymath}
    \abs{f_d^\phi(\vec{x}) - f_d^\phi(\vec{y})}^q \leq
    {\xi_{\phi,q}(k)}\norm[q]{\vec{x} - \vec{y}}^q.
  \end{displaymath}
\end{lemma}
\begin{proof}
  We apply Lemma~\ref{lem:mean-value}.  Assuming $\vec{z}$ is as
  defined in that lemma, we have by H\"{o}lder's inequality
{%
  \begin{align}
    \abs{f_d^\phi(\vec{x}) - f_d^\phi(\vec{y})} &\leq
    \Phi(f_d(\vec{z}))\sum_{i=1}^d\frac{\abs{y_i -
        x_i}}{\Phi(z_i)}\abs{\pdiff{f_d}{z_i}}\nonumber\\
    &\leq \Phi(f_d(\vec{z}))\inp{\sum_{i=1}^d
      \inp{\frac{1}{\Phi(z_i)}\abs{\pdiff{f_d}{z_i}}}^a}^{1/a}
      \norm[q]{\vec{x} - \vec{y}}.\nonumber
  \end{align}%
}%
Since $\phi$ is symmetrizable with exponent $a$, we can replace
$\vec{z}$ in the above inequality with a vector $\vec{\tilde{z}}$ all
of whose non-zero entries are equal to some fixed real $\tilde{z}$.
Let $k \leq d$ be the number of non-zero entries in $\vec{\tilde{z}}$.
Using the consistency condition, we then get
\begin{align}
  \abs{f_d^\phi(\vec{x}) - f_d^\phi(\vec{y})} &\leq
  \Phi(f_k(\tilde{z}))
  \inp{\sum_{i=1}^k
    \inp{
      \frac{1}{k\Phi(\tilde{z})}
      \abs{f_k'(\tilde{z})}}^a
  }^{1/a}
  \norm[q]{\vec{x} - \vec{y}}\nonumber\\
  &= \frac{1}{k^{1-1/a}}
  \frac{
    \Phi(f_k(\tilde{z}))\abs{f_k'(\tilde{z})}
  }{\Phi(\tilde{z})}
  \norm[q]{\vec{x} - \vec{y}}.\nonumber
\end{align}
Raising both sides to the power $q$, and using $\frac{1}{a} +
\frac{1}{q} = 1$ and the definitions of the functions $\Xi$ and $\xi$,
we get the claimed inequality.
\end{proof}
Given a message $\phi$ satisfying the conditions of Lemma\nobreakspace
\ref {lem:tech}, we can easily prove the following lemma on the
propagation of errors in locally finite infinite trees.  Recall
that $F_\rho(\sigma)$ denotes the value computed by the recurrence at
the root $\rho$ under an initial condition $\sigma$.  The lemma
quantifies the dependence of $F_\rho(\sigma)$ on initial conditions
$\sigma$ which are fixed everywhere except at some cutset $C$, in
terms of the distance of $C$ from $\rho$.
\begin{lemma}
  \label{lem:general-tree}
  Let $T$ be a finite tree rooted at $\rho$. Let $C$ be a cutset in
  $T$ at distance at least $1$ from the root which does not contain
  any leaves, and let $C'$ be the
  cutset consisting of the children of vertices in $C$. Consider two
  arbitrary initial conditions $\sigma$ and $\tau$ on $T_{\leq C'}$
  which differ only on $C'$, and which assign values from the interval
  $[0, b]$.  Given a recurrence family $\inb{f_d}_{d\geq 1}$ , let $a$
  and $q$ be positive reals such that $\frac{1}{a} + \frac{1}{q} = 1$
  and suppose $\phi$ is a message that is symmetrizable with exponent
  $a$. We then have
  \begin{equation*}
    |F_\rho(\sigma) - F_\rho(\tau)|^q \leq
    \inp{
      \frac{M}{L}
    }^q
    \sum_{v \in C} \alpha^{|v|},
  \end{equation*}
  where $\alpha$ is as defined in eq. (\ref{eq:8}), and $L$ and $M$
  are defined as follows:
  \begin{displaymath}
    L \defeq \inf_{x \in (0, b)} \phi'(x)\text{; \quad}
    M \defeq \max_{v \in C}\abs{\phi(F_v(\sigma)) - \phi(F_v(\tau))}.
  \end{displaymath}
\end{lemma}
For a proof of this lemma, see Appendix~\ref{sec:proof-lemma-refl}. %

\section{A message for the hard core model}
\label{sec:spec-mess-hard}

We now instantiate the approach outlined in
Section~\ref{sec:decay-corr-saw} to prove Theorems\nobreakspace \ref
{thm:main-hard-core} for the hard core model.  Our message is the same
as that used in~\cite{li_correlation_2011}; %
we choose
  \begin{equation}
    \phi(x) \defeq \sinh^{-1}\inp{\sqrt{x}}\text{, so that }\Phi(x) \defeq
    \phi'(x) = \frac{1}{2\sqrt{x(1+x)}}.\label{eq:4}
  \end{equation}%
Notice that $\phi$ is a strictly increasing, continuously
differentiable function on $(0, \infty)$, and also satisfies the
technical condition that the derivative $\Phi$ be bounded away from
zero on any finite interval, as required in the definition of a
message.  Our improvements over the results
in~\cite{sinclair13:_spatial} depend upon proving the following
fact about the message $\phi$.
\begin{lemma}
  \label{lem:observ-symm}
  For any $a \geq 2$, the message $\phi$ as defined in
  eq.\nobreakspace \textup {(\ref {eq:4})} is symmetrizable with
  exponent $a$ with respect to the tree recurrence
  $\inb{f_{d,\lambda}}_{d\geq 1}$ of the hard core model.
\end{lemma}
The proof of the above lemma is quite technical and is deferred to
Appendix~\ref{sec:symmt-mess}.  The crucial advantage of
Lemma~\ref{lem:observ-symm} is the flexibility it allows in the choice
of exponent $a$.  The decay factor $\alpha$ in
Lemma~\ref{lem:general-tree} which governs the rate of decay in the
tree recurrences depends upon the choice of the exponent, and as we show
later in this section, it is possible to obtain an optimal decay rate
by choosing an appropriate exponent satisfying
Lemma~\ref{lem:observ-symm}.  In contrast, the simpler
concavity arguments used in~\cite{sinclair13:_spatial} restricted the
choice of the exponent to $a = 2$, and hence led to a suboptimal decay rate.

We now show how to choose $a$ and $q$ so as to obtain an optimal decay
rate.  Our choice will depend upon the vertex activity $\lambda$, and
to clarify this dependence, we first define the quantity $\Delta_c
\defeq \Delta_c(\lambda)$ as the unique solution of $\lambda_c(t) =
\lambda$ (the existence and uniqueness of $\Delta_c$ follows from the
well known fact that $\lambda_c$ is a strictly decreasing function and
maps the interval $(1, \infty)$ onto $(0, \infty)$).  Our choice of
$a$ and $q$ will then enforce the following conditions:
\begin{enumerate}
\item $a \geq 2$, so that that $\phi$ is symmetrizable with exponent
  $a$ and hence Lemma~\ref{lem:general-tree} is applicable, and
\item For all $d > 0$, $\xi_{\phi,q}(d) \leq
  \frac{1}{\Delta_c}$, so that we get sufficient stepwise decay when
  Lemma~\ref{lem:general-tree} is applied.
\end{enumerate}
The second condition is the key to making the proof work, and in order
to enforce it we will need to analyze the function $\xi_{\phi,q}(d)$
in some detail.  We begin with the following simple lemma, the special
case $q=2$ of which was proven in~\cite{sinclair13:_spatial}. The
lemma merely shows how to perform one of the maximizations needed in
the definition of the decay factor $\alpha$. In what follows, we drop
the subscript $\phi$ for simplicity of notation.
\begin{lemma}
  \label{lem:fixed-point}
  Consider the hard core model with any fixed vertex activity $\lambda
  > 0$. For any $q \geq 1$ and with $\phi$ as defined in eq.\nobreakspace \textup {(\ref {eq:4})},
  we have $\xi_q(d) =
  \Xi_q(d, \tilde{x}_\lambda(d))$, where $\tilde{x}_\lambda(d)$ is the unique
  solution to
  \begin{equation}
    d\tilde{x}_\lambda(d) = 1 + f_{d,\lambda}(\tilde{x}_\lambda(d)).
  \end{equation}
\end{lemma}
\begin{proof}
Plugging in $\Phi$ from eq.\nobreakspace \textup {(\ref {eq:4})} in the definition of $\Xi$, we get
\begin{displaymath}
  \Xi_q(d, x) = d^{q-1}\inp{\frac{x}{1+x}\frac{f_{d,\lambda}(x)}{1+f_{d,\lambda}(x)}}^q.
\end{displaymath}
Taking the partial derivative with respect to the second argument, we
get
\begin{displaymath}
  \Xi_q^{(0,1)}(d,x) = \frac{q\Xi_q(d,x)}{2x(1+x)\inp{1+f_{d,\lambda}(x)}}
  \insq{1 + f_{d,\lambda}(x)  - dx}.
\end{displaymath}
For fixed $d$, the quantity outside the square brackets is always
positive, while the expression inside the square brackets is strictly
decreasing in $x$.  Thus, any zero of the expression in the brackets
will be a unique maximum of $\Xi_q$.  The fact that such a zero exists
follows by noting that the partial derivative is positive at $x = 0$
and negative as $x \rightarrow \infty$.  Thus, $\Xi_q(d, x)$ is
maximized at $\tilde{x}_{\lambda}(d)$ as defined above, and hence
$\xi_q(d) = \Xi_q(d, \tilde{x}_\lambda(d))$, as claimed.
\end{proof}
\noindent We now choose $a$ and $q$ as follows:
\begin{equation}
  \frac{1}{q} = 1 - \frac{\Delta_c - 1}{2}\log\inp{1+
    \frac{1}{\Delta_c - 1}}; \;\; \frac{1}{a} =  1  - \frac{1}{q}.\label{eq:17}
\end{equation}
Note that since $\log(1+y) \leq y$ for all $y \geq 0$, we get that $q
\leq 2$ (and hence $a \geq 2$) by using $y = \frac{1}{\Delta_c -
  1}$, and noting that $\Delta_c > 1$.  Thus, the first condition
above is already satisfied.  The values above are chosen to make sure
that the second condition is satisfied as well, as we prove in the next
lemma.  Now that the exponents $a$ and $q$ are fixed, we define the
function $\nu_\lambda(d)$ as follows in order to emphasize dependence upon
$\lambda$:
\begin{equation*}
  \nu_\lambda(d) \defeq \xi_q(d).
\end{equation*}
\begin{lemma}\label{lem:tau-props}
  Fix $\lambda > 0$ and let $\Delta_c(\lambda) > 1$ be the unique
  solution to $\lambda_c(t) = \lambda$. The function $\nu_\lambda :
  \mathbb{R}^+ \rightarrow \mathbb{R}^{+}$ is maximized at $d =
  \Delta_c \defeq \Delta_c(\lambda)$.
  Further, $$\nu_\lambda(\Delta_c(\lambda)) =
  \frac{1}{\Delta_c(\lambda)}.$$
\end{lemma}
The proof of the above lemma is somewhat technical, and is deferred to
Appendix~\ref{sec:monot-prop-nu_l}. The lemma shows that when $\lambda
< \lambda_c(\Delta)$, the decay factor $\alpha < \frac{1}{\Delta}$.
As we observed above, the main ingredient in the proof is the specific
choice of the exponent $a$ in eq.~(\ref{eq:17}), which in turn is
allowed only because of the flexibility in the choice of $a$ allowed
by Lemma~\ref{lem:observ-symm}.
{%
We now proceed with the proof of Theorem\nobreakspace
\ref{thm:main-hard-core}. This only requires some standard
arguments, with the only new ingredient being the improved estimate
on decay factor proved in Lemma~\ref{lem:tau-props}.
\begin{proof}[Proof of Theorem\nobreakspace \ref {thm:main-hard-core}]
  Let $\family{G}$ be any family of finite or infinite graphs with
  connective constant $\Delta$.  We prove the result for any fixed
  $\lambda$ such that $\lambda < \lambda_c(\Delta)$.  For such
  $\lambda$, we have $\Delta_c(\lambda) > \Delta$ (since $\lambda_c$
  is a decreasing function).  Using Lemma\nobreakspace \ref
  {lem:tau-props} we then see that there is an $\epsilon > 0$ such
  that $\nu_{\lambda}(d)\Delta \leq 1 - \epsilon$ for all $d > 0$.

  We first prove that the hard core model with these parameters
  exhibits strong spatial mixing on this family of graphs. Let $G$ be
  any graph from $\family{G}$, $v$ any vertex in $G$, and consider
  any boundary conditions $\sigma$ and $\tau$ on $G$ which differ only
  at a distance of at least $\ell$ from $v$.  We consider the Weitz
  self-avoiding walk tree $\Ts{v, G}$
  rooted at $v$ (as defined in Section~\ref{sec:connective-constant}).
   As before, we denote again by $\sigma$ (respectively, $\tau$)
   the translation of the boundary condition $\sigma$ (respectively,
   $\tau)$ on $G$ to $\Ts{v, G}$.  From Weitz's theorem, we then have that
  $R_v(\sigma, G) = R_v(\mathcal{W} \cup \sigma, \Ts{v, G})$ (respectively,
  $R_v(\tau, G) = R_v(\mathcal{W} \cup \tau, \Ts{v, G})$).

  Consider first the case where $G$ is infinite.  Let $C_\ell$ denote
  the cutset in $\Ts{v, G}$ consisting of all vertices at distance
  $\ell$ from $v$. Since $G$ has connective constant at most $\Delta$,
  it follows that for $\ell$ large enough, we have $|C_\ell| \leq
  \Delta^\ell(1-\epsilon/2)^{-\ell}$.  Further, in the notation of
  Lemma\nobreakspace \ref {lem:general-tree}, $\nu_\lambda(d)\Delta =
  1-\epsilon$ implies that the decay factor $\alpha$ (defined in
  eq. (\ref{eq:8})) is at most $(1-\epsilon)/\Delta$.  We now apply
  Lemma\nobreakspace \ref {lem:general-tree}.  We first observe that
  given our message $\phi$, we can bound the quantities $L$ and $M$ in
  the lemma as
  \begin{displaymath}
    L = \frac{1}{2\sqrt{\lambda(1+\lambda)}}\quad\text{, and }
    \quad M = \sinh^{-1}(\sqrt{\lambda}).
  \end{displaymath}
  The bounds on $L$ and $M$ follows from the fact that the values of
  the occupation ratio computed at any internal node of the tree lie
  in the range $[0, \lambda]$.  Setting $c_0 = (L/M)^q$, we can the
  apply the lemma to get {%
    \begin{align*}
      \abs{R_v(\sigma, G) - R_v(\tau, G)}^q &=
      \abs{R_v(\mathcal{W} \cup \sigma, \Ts{v,G})
        - R_v(\mathcal{W} \cup \tau, \Ts{v,G})}^q\\
      &\leq c_0\sum_{u \in
        C_\ell}\inp{\frac{1-\epsilon}{\Delta}}^\ell\\
      &\leq c_0\inp{\frac{1-\epsilon}{1-\epsilon/2}}^\ell,\text{ using
        $|C_\ell| \leq \Delta^\ell(1-\epsilon/2)^{-\ell}$,}
    \end{align*}%
    which establishes strong spatial mixing in $G$, since
    $1-\epsilon < 1-\epsilon/2$.
  }

  We now consider the case when $\family{G}$ is a family of finite
  graphs, and $G$ is a graph from $\family{G}$ of $n$ vertices.  Since
  the connective constant of the family is $\Delta$, there exist
  constants $a$ and $c$ (not depending upon $G$) such that for $\ell
  \geq a \log n$, $\sum_{i=1}^\ell N(v, \ell) \leq c\Delta^\ell$.  We
  now proceed with the same argument as in the infinite case, but
  choosing $\ell \geq a \log n$.  The cutset $C_\ell$ is again chosen
  to be the set of all vertices at distance $\ell$ from $v$ in
  $\Ts{v,G}$, so that $|C_\ell| \leq c\Delta^{\ell}$.  As before, we
  then have for $\ell > a \log n$, {%
    \begin{align}
      \abs{R_v(\sigma, G) - R_v(\tau, G)}^q &= \abs{R_v(\mathcal{W}
        \cup \sigma, \Ts{v,G})
        - R_v(\mathcal{W} \cup \tau, \Ts{v,G})}^q\nonumber\\
      &\leq c_0\sum_{u \in
        C_\ell}\inp{\frac{1-\epsilon}{\Delta}}^\ell\nonumber\\
      &\leq c\cdot c_0\inp{1-\epsilon}^{\ell}
      \text{, using $|C_\ell| \leq c\Delta^\ell$,}\label{eq:12}
    \end{align}%
    which establishes the requisite strong spatial mixing bound.
  }

{%
  In order to prove the algorithmic part, we first recall an observation of
  Weitz~\cite{Weitz06CountUptoThreshold} that an FPTAS for the
  ``non-occupation'' probabilities $1 - p_v$ under arbitrary boundary
  conditions is sufficient to derive an FPTAS for the partition
  function.  We further note that if the vertex $v$ is not already
  fixed by a boundary condition, then  $1-p_v = \frac{1}{1 + R_v} \geq
  \frac{1}{1+\lambda}$, since $R_v$ lies in the interval $[0,\lambda]$
  for any such vertex. Hence, an additive approximation to $R_v$ with
  error $\mu$ implies a multiplicative approximation to $1-p_v$ within
  a factor of $1\pm \mu(1+\lambda)$.  Thus, an algorithm that
  produces in time polynomial in $n$ and $1/\mu$ an \emph{additive}
  approximation to $R_v$ with error at most $\mu$ immediately gives an
  FPTAS for $1-p_v$, and hence, by Weitz's observation, also for the
  partition function.%
} %
  To derive such an algorithm, we again use the tree $\Ts{v, G}$ considered
  above.  Suppose we require an additive approximation with error at
  most $\mu$ to $R_v(\sigma, G) = R_v(\sigma, \Ts{v, G})$.  We notice
  first that $R_v = 0$ if and only if there is a neighbor of $v$ that
  is fixed to be occupied in the boundary condition $\sigma$.  In this
  case, we simply return $0$.  Otherwise, we expand $\Ts{v, G}$ up to
  depth $\ell$ for some $\ell \geq a\log n$ to be specified later.
  Notice that this subtree can be explored in time
  $O\inp{\sum_{i=1}^\ell N(v,i)}$ which is $O(\Delta^\ell)$ since the
  connective constant is at most $\Delta$.
  {%

  We
  now consider two extreme boundary conditions $\sigma_+$ and
  $\sigma_-$ on $C_\ell$: in $\sigma_+$ (respectively, $\sigma_-$) all
  vertices in $C_\ell$ that are not already fixed by $\sigma$ are
  fixed to ``occupied'' (respectively, unoccupied).  The form of the
  recurrence ensures that the true value $R_v(\sigma)$ lies between
  the values $R_v(\sigma_+)$ and $R_v(\sigma_-)$.  We compute the
  recurrence for both these boundary conditions on the tree.  The
  analysis leading to eq.\nobreakspace \textup {(\ref {eq:12})}
  ensures that, since $\ell \geq a \log n$, we have
  \begin{displaymath}
    \abs{R_v(\sigma_+, G) - R_v(\sigma_-, G)} \leq M_1\exp(-M_2\ell)
  \end{displaymath}
  for some fixed positive constants $M_1$ and $M_2$. Now, assume
  without loss of generality that $R_v(\sigma_+) \geq R_v(\sigma_-)$.
  By the preceding observations, we then have $$R_v(\sigma) \leq
  R_v(\sigma_+) \leq R_v(\sigma) + M_1\exp(-M_2\ell).$$ By choosing
  $\ell = a\log n + O(1) +O(\log (1/\mu))$, we get the required
  $\pm\mu$ approximation.  Further, by the observation above, the
  algorithm runs in time $O\inp{\Delta^{\ell}}$, which is polynomial in
  $n$ and $1/\mu$ as required.
}
\end{proof}

We now prove Corollary~\ref{cor:hard-core-gndn}.
\begin{proof}[Proof of Corollary~\ref{cor:hard-core-gndn}]
  Since $\lambda < \lambda_c(d)$, there exists an $\epsilon > 0$ such
  that $\lambda < \lambda_c(d(1+\epsilon))$.  Fix $\beta > 0$.  In
  order to prove the corollary, we only need to show that graphs drawn
  from $\G(n, d/n)$ have connective constant at most $d(1+\epsilon)$
  with probability at least $1 - n^{-\beta}$.

  Recall that $N(v,\ell)$ is the number of self-avoiding walks of
  length $\ell$ starting at $v$.  Suppose $\ell \geq a\log n$, where
  $a$ is a constant depending upon the parameters $\epsilon$, $\beta$
  and $d$ which will be specified later.  We first observe that
  \begin{displaymath}
    \E{\sum_{i=1}^\ell N(v,i)} \leq \sum_{i=1}^\ell\inp{\frac{d}{n}}^i n^i
    \leq d^{\ell}\frac{d}{d-1},
  \end{displaymath}
  and hence by Markov's inequality, we have $\sum_{i=1}^\ell N(v,i)
  \leq d^\ell\frac{d}{d-1}(1+\epsilon)^\ell$ with probability at
  least $1 - (1+\epsilon)^{-\ell}$.  By choosing $a$ such that
  $a\log(1+\epsilon) \geq \beta + 2$, we see that this probability
  is at least $1 - n^{-\inp{\beta+2}}$.  By taking a union bound over
  all $\ell$ with $a\log n \leq \ell \leq n$ and over all vertices
  $v$, we see that the connective constant $\Delta$ is at most
  $d(1+\epsilon)$ with probability at least $1-n^{-\beta}$.  We
  therefore see that with probability at least $1-n^{-\beta}$, the
  conditions of Theorem~\ref{thm:main-hard-core} are
  satisfied. This completes the proof.
\end{proof}

}%

\section{A message for the monomer-dimer model}
\label{sec:spec-mess-monom}

In this section, we apply the general framework of
Section~\ref{sec:decay-corr-saw} to the monomer-dimer model.  As in
the case of the hard core model, the first step is to choose an
appropriate message.  Unfortunately, unlike the case of the hard core
model where we could show that an already known message was
sufficient, we need to find a new message function in this case.  We
claim that the following message works:
\begin{equation}
  \phi(x) \defeq \frac{1}{2}\log\inp{\frac{x}{2-x}}
  \text{, so that } \Phi(x) \defeq
  \phi'(x) = \frac{1}{x(2-x)}.\label{eq:6}
\end{equation}
Note that $\phi$ is strictly increasing and continuously
differentiable on the interval $(0, 1]$, and its derivative is bounded
away from $0$ on that interval.  Thus, $\phi$ satisfies the conditions
required in the definition of a message (note that the bound $b$ used
in the definition is $1$ in the case of the monomer-dimer model).
Now, in order to apply Lemma~\ref{lem:general-tree}, we first study
the symmetrizability of $\phi$ in the following technical lemma.
\begin{lemma}
  \label{lem:observ-symm-monomer}
  Fix $r \in (1, 2]$. The message $\phi$ as defined in
  eq. (\ref{eq:6}) is symmetrizable with exponent $r$ with respect to
  the tree recurrences $\inb{f_{d,\gamma}}_{d \geq 1}$ of the
  monomer-dimer model.
\end{lemma}

We defer the proof of the above lemma to
Appendix~\ref{sec:symm-mess-monomer}.
As in the case of the hard core
model, we will need to make a careful choice of the exponent $r$ in order to
obtain an optimal decay factor. %
We begin with a technical lemma which characterizes the
behavior of the function $\xi$ used in the definition of the decay
factor.  For ease of notation, we drop the subscript $\phi$ from the
notation for $\xi$.
\begin{lemma}
  \label{lem:p-value}
  Consider the monomer-dimer model with edge activity $\gamma$, and
  let $\phi$ be the message chosen in (\ref{eq:6}).  For any $q > 1$,
  we have $\xi_q(d) = \Xi_q(d, \tilde{p}_\gamma(d))$, where
  $\tilde{p}_\gamma(d)$ satisfies $\Xi_q^{(0,1)}(d,
  \tilde{p}_\gamma(d)) = 0$ and is given by
  \begin{displaymath}
    \tilde{p}_\gamma(d) \defeq \frac{\sqrt{1+4\gamma d} - 1}{2\gamma d}.
  \end{displaymath}
\end{lemma}
\begin{proof}
  Plugging in $\Phi$ from eq.\nobreakspace \textup {(\ref {eq:6})} in
  the definition of $\Xi$, we get
\begin{displaymath}
  \Xi_q(d, x) = d^{q-1}
  \inp{
    \frac{\gamma x(2-x)f_{d,\gamma}(x)}{2-f_{d,\gamma}(x)}
  }^q =
  d^{q-1}\inp{
    \frac{\gamma x(2-x)}{1 + 2\gamma d x}
  }^q\text{, since $f_{d,\gamma}(x) = \frac{1}{1 + \gamma d x}$.}
\end{displaymath}
Taking the partial derivative with respect to the second argument, we
get
\begin{displaymath}
  \Xi_q^{(0,1)}(d,x) = \frac{2q\Xi_q(d,x)}{x(2-x)(1+2\gamma d x)}
  \insq{1 - x - \gamma d x^2}.
\end{displaymath}
For fixed $d$, and $0\leq x \leq 1$, the quantity outside the square
brackets is always positive, while the expression inside the square
brackets is strictly decreasing in $x$.  Thus, any zero of the
expression in the brackets in the interval $[0,1]$ will be a unique
maximum of $\Xi_q$.  By solving the quadratic, we see that
$\tilde{p}_\gamma(d)$ as defined above is such a solution.  Thus,
$\Xi_q(d, x)$ is maximized at $\tilde{p}_{\gamma}(d)$ as defined
above, and hence $\xi_q(d) = \Xi_q(d, \tilde{p}_\gamma(d))$.
\end{proof}

Given the edge activity $\gamma$ and an upper bound $\Delta$ on the
connective constant of the graph family being considered, we now
choose $D > \max(\Delta, 3/(4\gamma))$.  We claim that we can get
the required decay factor by choosing
\begin{equation}
  \frac{1}{r} = 1 - \frac{1}{\sqrt{1 + 4\gamma D}};
  \qquad \frac{1}{q} = 1 - \frac{1}{r} = \frac{1}{\sqrt{1+4\gamma D}}.\label{eq:11}
\end{equation}
Note that the choice of $D$ implies that $1 < r \leq 2$, so that
$\phi$ is symmetrizable with respect to $r$. The following lemma shows
that this choice of $r$ indeed gives us the required decay factor.  As
in the case of the hard core model, we emphasize the dependence of the
decay parameter on the model parameters by setting
\begin{displaymath}
  \nu_\gamma(d) \defeq \xi_q(d),
  \text{where $q$ is as chosen in eq.~(\ref{eq:11})}.
\end{displaymath}

\begin{lemma}\label{lem:tau-props-monomer}
  Fix $\gamma > 0$ and $D > 3/{4\gamma}$, and let $q$ be as chosen in
  (\ref{eq:11}). Then the function $\nu_\gamma : \mathbb{R}^+
  \rightarrow \mathbb{R}^{+}$ is maximized at $d = D$.
  Further, the decay factor $\alpha$ is given by
  \[
  \alpha = \nu_\gamma(D) = \frac{1}{D}
  \inp{
    1- \frac{2}{
      1 + \sqrt{1 + 4 \gamma D}
    }
  }^q.
  \]
\end{lemma}
\begin{proof}
  We consider the derivative of $\nu_\gamma(d)$ with respect to $d$.
  Recalling that $\nu_\gamma(d) = \xi(d)= \Xi(d,
  \tilde{p}_\gamma(d))$ and using the chain rule, we have
  \begin{align}
    \nu_\gamma'(d) &= \Xi^{(1,0)}(d,\tilde{p}) +
    \Xi^{(0,1)}(d,\tilde{p})\diff{\tilde{p}}{d}\nonumber\\
    &=\Xi^{(1,0)}(d,\tilde{p}), \text{ since $\Xi^{(0,1)}(d,\tilde{p})
      = 0$ by definition of $\tilde{p}$}\nonumber\\
    &=\frac{\Xi(d, \tilde{p})}{
      d(1+2\gamma d \tilde{p})
    }\insq{
      q - 1  - 2\gamma d \tilde{p}
    }%
    =\frac{\Xi(d, \tilde{p})}{
      d(1+2\gamma d \tilde{p})
    }\insq{
      \sqrt{1 + 4\gamma D} - \sqrt{1 + 4\gamma d}
    },\label{eq:13}
  \end{align}
  where we in the last line we substitute the values
  $\tilde{p}_\gamma(d) = (\sqrt{1 + 4\gamma d} - 1)/(2\gamma d)$ from
  Lemma~\ref{lem:p-value} and $q = \sqrt{1 + 4\gamma D}$ from
  eq.~(\ref{eq:11}).  Now, we note that in eq.~(\ref{eq:13}), the
  quantity outside the square brackets is always positive, while the
  quantity inside the square brackets is a strictly decreasing
  function of $d$ which is positive for $d < D$ and negative for $d >
  D$.  It follows that $\nu_\gamma'(d)$ has a unique zero at $d = D$
  for $d \geq 0$, and this zero is a global maximum of $\nu_\gamma$.
\end{proof}

We are now ready to prove our main result for the monomer-dimer model,
Theorem~\ref{thm:main-monomer-dimer} from the introduction.  Given
Lemmas~\ref{lem:general-tree} and \ref{lem:tau-props-monomer}, only
some standard arguments are needed to finish the proof.
\begin{proof}[Proof of Theorem~\ref{thm:main-monomer-dimer}]
  Let $\family{F}$ be any family of finite graphs with connective
  constant at most $\Delta$.  Given the vertex activity $\gamma$ of
  the monomer-dimer model, we choose $D = \max(\Delta, 3/(4\gamma))$.
  Using Lemma~\ref{lem:tau-props-monomer}, we then see that the decay
  factor $\alpha$ appearing in Lemma~\ref{lem:general-tree} can be
  chosen to be
  \begin{displaymath}
    \alpha = \frac{1}{D}\inp{1 - \frac{2}{1 + \sqrt{1 + 4\gamma D}}}^q.
  \end{displaymath}
  Now, let $G$ be any graph (with $n$ vertices) from $\family{F}$, and
  let $v$ be a vertex in $G$.  As observed in
  Section~\ref{sec:prob-likel-rati}, %
  it is sufficient to construct an
  FPTAS for $p_v(G)$ in order to derive an FPTAS for the partition function.

  Consider the self-avoiding walk tree $\Ts{v, G}$ rooted at $v$ (as
  defined in Section~\ref{thm:saw-tree}).  From Godsil's
  theorem~(Theorem~\ref{thm:saw-tree}), we know that $p_v(G) =
  p_v(\Ts{v, G})$.
  Let $C_\ell$ denote the cutset in $\Ts{v, G}$ consisting of all
  vertices at distance $\ell$ from $v$. Since $\mathcal{F}$ has connective
  constant at most $\Delta$, there exist constants $a$ and $c$ such
  that if $\ell \geq a \log n$, we have $\sum_{i=1}^\ell N(v, \ell)
  \leq c\Delta^\ell$.  We will now apply Lemma~\ref{lem:general-tree}
  with $q$ as defined in eq.~(\ref{eq:11}). We first observe that
  the quantities $L$ and $M$ in the lemma can be taken to be
  \begin{displaymath}
    L = 1,\quad\text{and,}\quad M = \frac{1}{2}\log (1 + 2\gamma n),
  \end{displaymath}
  since the degree of any
  vertex in $G$ is at most $n$.\footnote{Since the degree of every
    vertex $v$ in the graph is $n$, every boundary condition sigma
    satisfies $1 \geq F_v(\sigma) \geq \frac{1}{1+\gamma n}$.
    Substituting these bounds in the definition of $M$ in
    Lemma~\ref{lem:general-tree} yields the claimed bound.}
  Now, defining $c_0 \defeq (M/L)^q$, %
  we have %
    \begin{align}
      \abs{F_v(\vec{0}_\ell) - F_v(\vec{1}_\ell)}^q
      &\leq c_0 \sum_{u \in C_\ell} \alpha^\ell
      \leq c\cdot c_0\cdot\inp{\alpha \Delta}^{\ell}
      \text{, using $|C_\ell| \leq c\Delta^\ell$,}\nonumber\\
      &\leq c\cdot c_0 \cdot
      \inp{
        1 - \frac{2}{
          1 + \sqrt{1 + 4\gamma D
          }
        }
      }^{q\ell}
      \text{, using $D \geq \Delta$ after substituting for
        $\alpha$.}\label{eq:15}
    \end{align}%
    Raising both sides to the power $1/q$ and substituting for $c_0$ and $q$, we
    then have
    \begin{equation}
      \abs{F_v(\vec{0}_\ell) - F_v(\vec{1}_\ell)} \leq
      \frac{1}{2} c^{1/\sqrt{1 + 4\gamma D}} \cdot \log (1+2\gamma n) \cdot
      \inp{
        1 - \frac{2}{
          1 + \sqrt{1 + 4\gamma D
          }
        }
      }^{\ell}.\label{eq:16}
    \end{equation}
  To analyze the running time, we note that in
  order to obtain a $(1 \pm \epsilon)$ multiplicative approximation to
  $p_v(G)$, it is sufficient to obtain a $\pm\epsilon/(1 +\gamma n)$
  additive approximation; this is because $p_v(G) \geq 1/(1 +\gamma
  n))$ since the degree of each vertex in $G$ is at most $n$.  Now, as
  observed in Section~\ref{sec:trunc-recurr-with},
  $p_v(G)$ always lies between the quantities $F_v(\vec{0}_\ell)$ and
  $F_v(\vec{1}_\ell)$, so in order to obtain a $\pm\epsilon/(1+\gamma n)$
  approximation, it is sufficient to choose $\ell \geq a \log n$
  large enough so that the right hand side of eq.~(\ref{eq:16}) is
  at most $\epsilon/(1 + \gamma n)$.
  Denoting by $\beta$ the quantity in the parenthesis on the right hand
  side of eq.~(\ref{eq:16}), we can ensure this by choosing
  \begin{displaymath}
    \ell \geq \frac{1}{\log (1/\beta)}
    \insq{
      \log\frac{1+\gamma n}{\epsilon} +
      \log \log\inp{\sqrt{1 + 2\gamma n}} + \frac{1}{
        \sqrt{1 + 4\gamma D}
      }
      \log c
    }.
  \end{displaymath}
  Further, given such an $\ell$, the running time of the algorithm is
  $O(\sum_{i=1}^\ell N(v, \ell)) = O(\Delta^\ell)$, since this is the
  time it takes to expand the self-avoiding walk tree up to depth $\ell$.
  Noting that $1/(\log(1/\beta)) = \sqrt{\gamma \Delta} + \Theta(1)$, we
  obtain an algorithm running in time
  \begin{displaymath}
    ((1+ \gamma n)/\epsilon)^{
      O(\sqrt{\gamma\Delta}\cdot \log\Delta)
    }
  \end{displaymath}
  which provides a $(1\pm\epsilon)$ multiplicative approximation
  for $p_v(G)$.  Recalling the arguments in
  Section~\ref{sec:from-prob-part}, we see that this yields an
  algorithm for approximating the partition function up to a
  multiplicative factor of $(1 \pm \epsilon)$ with the same asymptotic
  exponent in the running time. This completes the proof.
\end{proof}
{
\begin{remark}
  Note that eq.~(\ref{eq:16}) can be interpreted as showing that (when
  $\Delta > 1/(4\gamma)$), strong spatial mixing holds with rate $1 -
  \frac{2}{1 + \sqrt{1+ 4\gamma \Delta}}$ on graphs of connective
  constant at most $\Delta$, and the factor $\sqrt{\gamma\Delta}$ in
  the exponent of our runtime is a direct consequence of this fact (in
  particular, strong spatial mixing at rate $c$ translates into the
  exponent being proportional to $\log c$).  Recall that Bayati
  \emph{et al.}~\cite{bayati_simple_2007} obtained an algorithm with
  the same runtime but only for graphs with \emph{maximum degree}
  $\Delta +1$ (which is a strict subset of the class of graphs with
  connective constant $\Delta$).  This was due to the fact that their
  analysis essentially amounted to proving that eq.~(\ref{eq:16}) hold
  for the special case of graphs of maximum degree $\Delta + 1$.
  Using an observation of Kahn and Kim~\cite{kahn_random_1998} that
  the rate of spatial mixing on the infinite $d$-ary tree is $1 -
  1/\Theta(\sqrt{\gamma\Delta})$, they concluded that such a runtime
  was the best possible for algorithms that use decay of correlations
  in this direct fashion.  We note here that since the infinite
  $d$-ary tree also has connective constant exactly $d$, this
  observation also implies that the rate of strong spatial mixing
  obtained in eq.~(\ref{eq:16}) is optimal for graphs of connective
  constant $\Delta$ (in fact, the rate of strong spatial mixing on the
  $d$-ary tree is exactly $1-\frac{2}{1 + \sqrt{1+ 4\gamma d}}$).
\end{remark}
}

\section{Branching factor and uniqueness of the Gibbs measure on
  general trees}
\label{sec:more-results-trees}
\label{sec:branch-numb-uniq}
We close with an application of our results to finding thresholds for
the uniqueness of the Gibbs measure of the hard core model on locally
finite infinite trees.  Our bounds will be stated in terms of the
\emph{branching factor}, which has been shown to be the appropriate
parameter for establishing phase transition thresholds for symmetric
models such as the ferromagnetic Ising, Potts and Heisenberg
models~\cite{lyons_ising_1989,pemantle_robust_1999}.
We begin with a general definition of the notion of uniqueness of
Gibbs measure (see, for example, the survey article of
Mossel~\cite{mossel_survey:_2004}).  Let $T$ be a locally finite
infinite tree rooted at $\rho$, and let $C$ be a cutset in $T$.
Consider the hard core model with vertex activity $\lambda >0$ on $T$.
We define the \emph{discrepancy} $\delta(C)$ of $C$ as follows.  Let
$\sigma$ and $\tau$ be boundary conditions in $T$ which fix the state
of the vertices on $C$, but not of any vertex in $T_{<C}$.  Then,
$\delta(C)$ is the maximum over all such $\sigma$ and $\tau$ of the
quantity $R_\rho(\sigma, T) - R_\rho(\tau,T)$.
\begin{definition}
  \textbf{(Uniqueness of Gibbs measure).} The hard core model with
  vertex activity $\lambda > 0$ is said to exhibit \emph{uniqueness of
    Gibbs measure} on $T$ if there exists a sequence of cutsets
  $\inp{B_i}_{i=1}^\infty$ such that $\lim_{i\rightarrow\infty}
  d(\rho, B_i) \rightarrow \infty$ and such that
  $\lim_{i\rightarrow\infty}\delta(B_i) = 0$.
\end{definition}
\begin{remark}
  Our definition of \emph{uniqueness} here is similar in form to those
  used by Lyons~\cite{lyons_ising_1989} and Pemantle and Steif~\cite{pemantle_robust_1999}.  Notice, however, that the recurrence
  for the hard core model implies that the discrepancy is
  ``monotonic'' in the sense that if cutsets $C$ and $D$ are such that
  $C < D$ (i.e., every vertex in $D$ is the descendant of some vertex
  in $C$) then $\delta(C) > \delta(D)$.  This ensures that the choice
  of the sequence $\inp{B_i}_{i=1}^\infty$ in the definition above is
  immaterial.  For example, uniqueness is defined by
  Mossel~\cite{mossel_survey:_2004} in terms of the cutsets $C_\ell$
  consisting of vertices at distance exactly $\ell$ from the root.
  However, the above observation shows that for the hard core model,
  Mossel's definition is equivalent to the one presented here.
\end{remark}

We now define the notion of the branching factor of an infinite tree.
\begin{definition}
  \textbf{(Branching
    factor~\cite{lyons_ising_1989,lyons_random_1990,pemantle_robust_1999}).}
  Let $T$ be an infinite tree.  The branching factor $\br{T}$ is
  defined as follows:
  \begin{equation*}
    \br{T} \defeq \inf\inb{b > 0\left| \inf_C\sum_{v\in C}b^{-|v|} = 0
      \right.},
  \end{equation*}
  where the second infimum is taken over all cutsets $C$.
\end{definition}
To clarify this definition, we consider some examples.  If $T$ is a
$d$-ary tree, then $\br{T} = d$.  Further, by taking the second
infimum over the cutsets $C_\ell$ of vertices at distance $\ell$ from
the root, it is easy to see that the branching factor is never more
than the connective constant.  Further, Lyons~\cite{lyons_random_1990}
observes that in the case of spherically symmetric trees, one can define
the branching factor as
$\liminf_{\ell\rightarrow\infty}N(\rho,\ell)^{\ell}$.

We are now ready to state and prove our results on the uniqueness of
the hard core model on general trees.
\begin{theorem}\label{thm:uniq-tree}
  Let $T$ be an infinite tree rooted at $\rho$ with branching factor
  $b$.  The hard core model with vertex activity $\lambda > 0$
  exhibits uniqueness of Gibbs measure on $T$ if $\lambda <
   \lambda_c(b)$.
\end{theorem}

\begin{proof}
  As before, we apply Lemma ~\ref{lem:general-tree} specialized to the
  message in eq.~(\ref{eq:4}) with the values of $a$ and $q$ as chosen
  in eq. (\ref{eq:17}), where as before $\Delta_c$ is the unique
  positive real number satisfying $\lambda = \lambda_c(b)$. Since
  $\lambda < \lambda_c(b)$, we then see from
  Lemma~\ref{lem:tau-props} that the decay factor $\alpha <
  \frac{1}{b}$.  Hence, there is an $\epsilon > 0$ such that $\alpha =
  \frac{1}{b(1+\epsilon)}$ for some $\epsilon > 0$.  Applying
  Lemma~\ref{lem:general-tree} to an arbitrary cutset $C$ as in the
  proof of Theorem~\ref{thm:main-hard-core}, we then get
  \begin{equation}
    \delta(C)^q \leq M_0\sum_{v \in C}\insq{(1+\epsilon)b}^{-|v|},\label{eq:25}
  \end{equation}
  where $M_0$ is a constant.  Since $b(1+\epsilon) > \br{T}$, the definition of $\br{T}$ implies
  that we can find a sequence $\inp{B_i}_{i=1}^\infty$ of cutsets such
  that
  \begin{equation*}
    \lim_{i\rightarrow\infty}\sum_{v \in
      B_i}\insq{(1+\epsilon)b}^{-|v|} = 0.%
  \end{equation*}
  Further, such a sequence must satisfy
  $\lim_{i\rightarrow\infty}d(\rho, B_i) = \infty$, since otherwise
  the limit above would be positive.  Combining with
  eq.~(\ref{eq:25}), this shows that
  $\lim_{i\rightarrow\infty}\delta(B_i) = 0$, which completes the
  proof.
\end{proof}

\section{Concluding remarks}
Our results show that the connective constant is a natural notion of
``average degree'' of a graph that captures the spatial mixing
property of the hard core and monomer-dimer models on the graph.  We
can then ask if such a correspondence holds also for other spin
systems. More precisely, we start with results that establish decay of
correlations in a graph when the parameters of the spin system lie in
a region~\cite{sinclair_approximation_2012, li_correlation_2011,
  Weitz06CountUptoThreshold}, say $R(d)$, determined only by the
maximum degree $d$ of the graph, and ask whether we can claim that the
maximum degree $d$ in those results can be replaced by $\Delta+1$,
where $\Delta$ is the connective constant of the graph
(Theorem~\ref{thm:main-hard-core}, for example, achieves precisely
this goal for the hard core model).

It turns out that the only other spin system for which we can prove an
exact analogue of Theorem~\ref{thm:main-hard-core} is the zero field Ising
model, (we omit the proof, but refer to Lyons~\cite{lyons_ising_1989},
who proved a similar result for the uniqueness of Gibbs measure of the
zero field Ising model on arbitrary trees).  However, it seems that
for the anti-ferromagnetic Ising model with field, one cannot hope to
prove an exact analog of Theorem~\ref{thm:main-hard-core}. (It is, however,
possible to obtain a weaker result which is analogous to Theorem 1.3
in \cite{sinclair13:_spatial}, in the sense that its hypotheses
require an upper bound on the maximum degree as well.)

Pursuant to the above observations, we are thus led to ask if there is
a natural notion of ``average degree''---different from the connective
constant---using which one can show decay of correlations for the
anti-ferromagnetic Ising model with field on graphs of unbounded
degree.  More precisely, we would want such a notion to be powerful
enough to yield an analog of Corollary~\ref{cor:hard-core-gndn} for the
anti-ferromagnetic Ising model with field on random sparse graphs.

The problem of identifying conditions under which a given spin system
exhibits uniqueness of Gibbs measure on a specific lattices also has
received much attention in the statistical physics literature.  For
example, it was predicted by Gaunt and
Fisher~\cite{gaunt_hardsphere_1965} and by Baxter, Enting and
Tsang~\cite{baxter_hard-square_1980} that the hard core model on
$\Z^2$ should exhibit uniqueness of Gibbs measure as long as $\lambda
< 3.79$.  However, as we pointed out above, the condition $\lambda <
1.6875$ obtained by Weitz~\cite{Weitz06CountUptoThreshold} was the
best known rigorous upper bound until the work of Restrepo \emph{et
  al.}~\cite{restrepo11:_improv_mixin_condit_grid_count} and Vera
\emph{et al.}~\cite{vera13:_improv_bound_phase_trans_hard}, and these
latter papers also employed Weitz's idea of analyzing the
self-avoiding walk (SAW) tree. Although the analysis in this paper is not as
strictly specialized to the case of $\Z^2$ as were those of Restrepo
\emph{et al.} and Vera \emph{et al.}, it still follows Weitz's
paradigm of analyzing a SAW tree.

A natural question therefore arises: can the method of analyzing decay
of correlations on the SAW tree establish uniqueness of Gibbs measure
of the hard core model on $\Z^2$ under the condition $\lambda < 3.79$
predicted in ~\cite{gaunt_hardsphere_1965} and
\cite{baxter_hard-square_1980}?  One reason to suspect that this might
not be the case is that the SAW tree approach typically establishes
not only that uniqueness of the Gibbs measure holds on the lattice but
also that \emph{strong} spatial mixing holds on the SAW tree---the
latter condition has algorithmic implications and is potentially
stronger than the uniqueness of the Gibbs measure.
It is in keeping with this
suspicion that the best bounds we obtain for $\Z^2$ are still quite
far from the conjectured bound of $\lambda < 3.79$.  This contrast
becomes even more pronounced when we consider the triangular lattice
$\mathbb{T}$: Baxter~\cite{baxter_hard_1980} solved the hard core
model on $\mathbb{T}$ exactly (parts of Baxter's solution were later
made rigorous by Andrews~\cite{andrews_hard-hexagon_1981}) using
combinatorial tools, and showed that the uniqueness of Gibbs measure
holds when $\lambda < 11.09$.  On the other hand, the SAW tree
approach---using published estimated of the connective estimates of
$\mathbb{T}$ as a black box---can only show so far that strong spatial
mixing holds under the condition $\lambda < 0.961$ (see
Table~\ref{fig:1}).  However, it must be pointed out that it may be
possible to significantly improve this bound by analyzing the
connective constant of the Weitz SAW tree of $\mathbb{T}$, as done in
Appendix~\ref{sec:descr-numer-results} for $\Z^2$.

\paragraph{Acknowledgments.} We thank Elchanan Mossel, Allan Sly, Eric
Vigoda and Dror Weitz for helpful discussions.%

{

}

\appendix

\section{Description of numerical results}
\label{sec:descr-numer-results}

In this section, we describe the derivation of the numerical bounds in
Table~\ref{fig:1}.  As in~\cite{sinclair13:_spatial}, all of the
bounds are direct applications of Theorem~\ref{thm:main-hard-core}
using published upper bounds on the connective constant for the
appropriate graph (except for the starred bound of 2.538 for the case of
$\mathbb{Z}^2$, which we discuss in greater detail below).  The exact
connective constant is not known for the Cartesian lattices $\Z^2,
\Z^3, \Z^4, \Z^5$ and $\Z^6$, and the triangular lattice~$\mathbb{T}$,
and we use the rigorous upper and lower bounds available in the
literature~\cite{madras96:_self_avoid_walk,
  weisstein:_self_avoid_walk_connec_const}.  In contrast, for the
honeycomb lattice, Duminil-Copin and
Smirnov~\cite{duminil-copin_connective_2012} rigorously established
the connective constant to be $\mathbb{H}$ is $\sqrt{2 + \sqrt{2}}$ in
a recent breakthrough, and this is the bound we use for that lattice.
In order to apply Theorem~\ref{thm:main-hard-core} for a given lattice
of connective constant at most $\Delta$, we simply need to compute
$\lambda_c(\Delta) = \frac{\Delta^\Delta}{(\Delta-1)^{(\Delta+1)}}$,
and the monotonicity of $\lambda_c$ guarantees that the lattice
exhibits strong spatial mixing as long as $\lambda <
\lambda_c(\Delta)$.

We now consider the special case of $\mathbb{Z}^2$. As we pointed out
in the introduction, any improvement in the connective constant of a
lattice (or that of the Weitz SAW tree corresponding to the lattice)
will immediately lead to an improvement in our bounds.  In fact, as we
discuss below, Weitz's construction allows for significant
freedom in the choice of the SAW tree.  We show here that using
a tighter combinatorial analysis of the connective constant of a
suitably chosen Weitz SAW tree of $\mathbb{Z}^2$, we can improve upon
the bounds obtained by Restrepo \emph{et
al.}~\cite{restrepo11:_improv_mixin_condit_grid_count} and Vera
\emph{et al.}~\cite{vera13:_improv_bound_phase_trans_hard} using
sophistical methods tailored to the special case of $\mathbb{Z}^2$.
Our basic idea is to exploit the fact that the Weitz SAW tree adds
additional boundary conditions to the canonical SAW tree of the
lattice.  Thus, it allows a strictly smaller number of self-avoiding
walks than the canonical SAW tree, and therefore can have a smaller
connective constant than that of the lattice itself.  Further, as
in~\cite{sinclair13:_spatial}, the proof of
Theorem~\ref{thm:main-hard-core} only uses the Weitz SAW tree, and
hence the bounds obtained there clearly hold if the connective
constant of the Weitz SAW tree is used in place of the connective
constant of the lattice.

The freedom in the choice of the Weitz SAW tree---briefly alluded to
above---also offers the opportunity to incorporate another tweak which
can potentially increase the effect of the boundary conditions on the
connective constant.  In Weitz's construction, the boundary conditions
on the SAW tree are obtained in the following way (see Theorem 3.1 in
\cite{Weitz06CountUptoThreshold}).  First, the neighbors of each
vertex are ordered in a completely arbitrary fashion: this ordering
need not even be consistent across vertices.  Whenever a loop, say
$v_0, v_1, \ldots, v_l, v_0$ is encountered in the construction of the
SAW tree, the occurrence of $v_0$ which closes the loop is added to
the tree along with a boundary condition which is determined by the
ordering at $v_0$: if the neighbor $v_1$ (which ``started'' the loop)
happens to be smaller than $v_l$ (the last vertex before the loop is
discovered) in the ordering, then the last copy of $v_0$ appears in
the tree fixed as ``occupied'', while otherwise, it appears as
``unoccupied''.

The orderings at the vertices need not even be fixed in advance, and
different copies of the vertex $v$ appearing in the SAW tree can have
different orderings, as long as the ordering at a vertex $v$ in the
tree is a function only of the path from the root of the tree to $v$.
We now specialize our discussion to $\Z^2$.  The simplest such
ordering is the ``uniform ordering'', where we put an ordering on the
cardinal directions \emph{north}, \emph{south}, \emph{east} and
\emph{west}, and order the neighbors at each vertex in accordance with
this ordering on the directions.  This was the approach used by
Restrepo \emph{et
  al}.~\cite{restrepo11:_improv_mixin_condit_grid_count}.

However, it seems intuitively clear that it should be possible to
eliminate more vertices in the tree by allowing the ordering at a
vertex $v$ in the tree to depend upon the path taken from the origin
to $v$.  We use a simple implementation of this idea by using a
``relative ordering'' which depends only upon the last step of such a
path.  In particular, there are only three possible options available
at a vertex $v$ in the tree (except the root): assuming the parent of
$v$ in the tree is $u$: the first is to go \emph{straight}, i.e., to
proceed to the neighbor of $v$ (viewed as a point in $\mathbb{Z}^2$
which lies in the same direction as the vector $v - u$, where $v$ and
$u$ are again viwed as points in $\mathbb{Z}^2$).  Analogously, we can
also turn \emph{left} or \emph{right} with respect to this direction.
Our ordering simply stipulates that \emph{straight} $>$ \emph{right}
$>$ \emph{left}.

To upper bound the connective constant of the Weitz SAW tree, we use
the standard method of \emph{finite memory self-avoiding
  walks}~\cite{madras96:_self_avoid_walk}---these are walks which are
constrained only to not have cycles of length up to some finite length
$L$. Clearly, the number of such walks of any given length $\ell$
upper bounds $N(v, \ell)$.  In order to bring the boundary conditions
on the Weitz SAW tree into play, we further enforce the constraint
that the walk is not allowed to make any moves which will land it in a
vertex fixed to be ``unoccupied'' by Weitz's boundary conditions (note
that a vertex $u$ can be fixed to be ``unoccupied'' also because one
of its children is fixed to be ``occupied'': the independence set
constraint forces $u$ itself to be ``unoccupied'' in this case, and
hence leads to additional pruning of the tree by allowing the other
children of $u$ to be ignored).  Such a walk can be in one of a finite
number $k$ (depending upon $L$) of states, such that the number of
possible moves it can make to state $j$ while respecting the above
constraints is some finite number $M_{ij}$.  The $k\times k$ matrix $M
= (M_{ij})_{i,j\in[k]}$ is called the \emph{branching
  matrix}~\cite{restrepo11:_improv_mixin_condit_grid_count}. We
therefore get $N(v,\ell) \leq \vec{e_1}^TM^\ell\vec{1}$, where
$\vec{1}$ denotes the all $1$'s vector, and $\vec{e_1}$ denotes the
co-ordinate vector for the state of the zero-length walk.

Since the entries of $M$ are non-negative, the Perron-Frobenius
theorem implies that one of the maximum magnitude eigenvalues of
the matrix $M$ is a positive real number $\gamma$.  Using Gelfand's
formula (which states that $\gamma =
\lim_{\ell\rightarrow\infty}\norm{M^\ell}^{1/\ell}$, for any fixed
matrix norm) with the $\ell_\infty$ norm to get the last equality, we
see that
\begin{displaymath}
  \limsup_{\ell\rightarrow\infty} N(v,\ell)^{1/\ell} \leq
  \limsup_{\ell\rightarrow\infty}(\vec{e_1}^TM^\ell\vec{1})^{1/\ell}
  \leq \limsup_{\ell\rightarrow\infty}\norm[\infty]{M^\ell}^{1/\ell}
  = \gamma.
\end{displaymath}
Hence, the largest real eigenvalue $\gamma$ of $M$ gives a bound on
the connective constant of the Weitz SAW tree.

Using the matrix $M$ corresponding to walks in which cycles of length
at most $L=26$ are avoided, we get that the connective constant of the
Weitz SAW tree is at most $2.433$ (we explicitly construct the matrix
$M$ and then use Matlab to compute its largest eigenvalue). Using this
bound for $\Delta$, and applying Theorem~\ref{thm:main-hard-core} as
described above, we get the bound $2.529$ for $\lambda$ in the
notation of the table, which is better than the bounds obtained by
Restrepo \emph{et
  al.}~\cite{restrepo11:_improv_mixin_condit_grid_count} and Vera
\emph{et al.}~\cite{vera13:_improv_bound_phase_trans_hard}.  With
additional computational optimizations we can go further and analyze
self avoiding walks avoiding cycles of length at most $L=30$. The
first optimization is merging ``isomorphic'' states (this will
decrease the number of states and hence the size of $M$ significantly,
allowing computation of the largest eigenvalue): formally, the state
of a SAW will be a suffix of length $s$ such that the Manhattan
distance between the final point and the point $s$ steps in the past
is less than $L-s$ (note that the state of a vertex in the SAW tree
can be determined from the state of its parent and the last step), and
two states are isomorphic if they have the same neighbors at the next
step of the walk. The second optimization is computing the largest
eigenvalue using the power method. For $L=30$ we obtain that the
connective constant of the Weitz SAW tree is at most $2.429$, which on
applying Theorem~\ref{thm:main-hard-core} yields the bound $2.538$ for
$\lambda$, as quoted in Table~\ref{fig:1}.

\section{Proofs omitted from Section~\ref{sec:decay-corr-saw}}
\label{sec:proof-lemma-refl}
We include here a proof of Lemma~~\ref{lem:mean-value}
for the convenience of the reader.
\begin{proof}[Proof of Lemma~\ref{lem:mean-value}]
  Define $H(t) \defeq f_{d,\lambda}^\phi(t\vec{x} + (1-t)\vec{y})$ for $t \in
  [0,1]$.  By the scalar mean value theorem applied to $H$, we have
  \begin{equation*}
    f_{d,\lambda}^\phi(\vec{x}) - f_{d,\lambda}^\phi(\vec{y}) = H(1) -
    H(0) = H'(s)\text{, for some $s \in [0,1]$}.
  \end{equation*}
  Let $\psi$ denote the inverse of the message $\phi$: the derivative
  of $\psi$ is given by $\psi'(y) = \frac{1}{\Phi(\psi(y))}$,
  where $\Phi$ is the derivative of $\phi$.  We now define the vector
  $\vec{z}$ by setting $z_i = \psi(sx_i + (1-s)y_i)$ for $1 \leq i
  \leq d$.  We then have
  \begin{align}
    \abs{f_{d,\lambda}^\phi(\vec{x}) - f_{d,\lambda}^\phi(\vec{y})} &=
    \abs{H'(s)} = \abs{\ina{\nabla{f_{\lambda, d}^\phi(s\vec{x} +
        (1-s)\vec{y})},
        \vec{x} - \vec{y}}}\nonumber\\
    &= \Phi(f_{d,\lambda}(\vec{z}))
    \abs{
      \sum_{i=1}^d
      \frac{x_i - y_i}{\Phi(z_i)}
      \pdiff{f_{d,\lambda}}{z_i}
    }, \quad\text{using the chain rule}\nonumber\\
    &\leq
    \Phi\inp{f_{d,\lambda}(\vec{z})}
    \sum_{i=1}^d
    \frac{\abs{y_i-x_i}}{\Phi(z_i)}
    \abs{\pdiff{f_{d,\lambda}}{z_i}},
    \quad\text{as claimed.}\nonumber
  \end{align}
  We recall that for simplicity, we are using here the somewhat
  non-standard notation $\pdiff{f}{z_i}$ for the value of the partial
  derivative $\pdiff{f}{R_i}$ at the point $\vec{R} = \vec{z}$.
\end{proof}

We now give the proof of the Lemma~\ref{lem:general-tree}.  The proof
is syntactically identical to the proof of a similar lemma in
\cite{sinclair13:_spatial}, and the only difference (which is of
course crucial for our purposes) is the use of the more specialized
Lemma~\ref{lem:tech} in the inductive step.  %
\begin{proof}[Proof of Lemma~\ref{lem:general-tree}]
  Recall that given a vertex $v$ in $T_{\leq C}$, $T_v$ is the subtree
  rooted at $v$ and containing all the descendants of $v$, and
  $F_v(\sigma)$ is the value computed by the recurrence at the root
  $v$ of $T_v$ under the initial condition $\sigma$ restricted to
  $T_v$.  We will denote by $C_v$ the restriction of the cutset $C$ to
  $T_v$.

  By induction on the structure of $T_\rho$, we will now show that for
  any vertex $v$ in $T_\rho$ which is at a distance $\delta_v$ from
  $\rho$, and has arity $d_v$, one has
  \begin{equation}\label{eq:7}
    |\phi(F_v(\sigma)) - \phi(F_v(\tau))|^q
    \leq M^q\sum_{u \in C_v} \alpha^{|u|-\delta_v}.
  \end{equation}
  To see that this implies the claim of the lemma, we observe that
  since $F_\rho(\sigma)$ and $F_\rho(\tau)$ are in the interval $[0,
  b]$, we have $|F_v(\sigma) - F_v(\tau)| \leq
  \frac{1}{L}|\phi(F_v(\sigma)) - \phi(F_v(\tau))|$.  Hence, taking $v
  = \rho$ in eq.~(\ref{eq:7}),  the claim of the lemma follows from
  the above observation.

  We now proceed to prove eq.~(\ref{eq:7}). The base case of the
  induction consists of vertices $v$ which are either of arity $0$ or
  which are in $C$.  In the first case (which includes the case where
  $v$ is fixed by both the initial conditions to the same value), we
  clearly have $F_v(\sigma) = F_v(\tau)$, and hence the claim is
  trivially true.  In the second case, we have $C_v = \inb{v}$, and
  all the children of $v$ must lie in $C'$.  Thus, in this case, the
  claim is true by the definition of $M$.

  We now proceed to the inductive case.  Let $v_1, v_2, \ldots
  v_{d_v}$ be the children of $v$, which satisfy eq.~(\ref{eq:7}) by
  induction.  In the remainder of the proof, we suppress the
  dependence of $\xi$ on $\phi$ and $q$. Applying Lemma~\ref{lem:tech}
  followed by the induction hypothesis, we then have, for some
  positive integer $k\leq d_v$
  \begin{align*}
    \abs{\phi(R_v(\sigma)) - \phi(R_v(\tau))}^q &\leq
    \xi(k)\sum_{i=1}^{d_v}\abs{\phi(R_{v_i}(\sigma)) -
      \phi(R_{v_i}(\tau))}^q\text{, using Lemma~\ref{lem:tech}}\\
    &\leq M^q\xi(k)\sum_{i=1}^{d_v}\sum_{u \in C_{v_i}}\alpha^{|u| -
      \delta_{v_i}}\text{, using the
      induction hypothesis}\\
    &\leq M^q\sum_{u \in C_v}\alpha^{|u| - \delta_{v}}\text{,
      using $\xi(k) \leq \alpha$ and $\delta_{v_i} = \delta_{v}
      + 1$}.
  \end{align*}
  This completes the induction.
\end{proof}

\section{Proofs omitted from Section~\ref{sec:spec-mess-hard}}
\label{sec:proofs-omitted-from-4}
\subsection{Maximum of $\nu_\lambda$ and implications for strong
  spatial mixing}
\label{sec:monot-prop-nu_l}
In this section, we prove Lemma~\ref{lem:tau-props}.
\begin{proof}[Proof of Lemma~\ref{lem:tau-props}]
  We first prove that given $\lambda$, $\tilde{x}_\lambda(d)$ is a
  \emph{decreasing} function of $d$.  For ease of notation, we
  suppress the dependence of $\tilde{x}_\lambda(d)$ on $d$ and
  $\lambda$.  From Lemma~\ref{lem:fixed-point}, we know that
  $\tilde{x}$ is the unique positive solution of $d\tilde{x} = 1  +
  f_d(\tilde{x})$.  Differentiating the equation with respect to
  $d$ (and denoting $\diff{\tilde{x}}{d}$ by $\tilde{x}'$), we have
  \begin{displaymath}
    \tilde{x} + d\tilde{x}' = -f_d(\tilde{x})\insq{\frac{d\tilde{x}'}{1+\tilde{x}}  + \log(1+\tilde{x})}
  \end{displaymath}
  which in turn yields
  \begin{align}
    \tilde{x}' &=
    -\frac{(1+\tilde{x})\insq{f_d(\tilde{x})\log(1+\tilde{x}) +
        \tilde{x}}}{d(1+d)\tilde{x}}.\label{eq:21}
  \end{align}
  Since $\tilde{x} \geq 0$, this shows that $\tilde{x}$ is a decreasing
  function of $d$.

  We now consider the derivative of $\nu_\lambda(d)$ with respect to
  $d$.  Recalling that $\nu_\lambda (d) = \xi(d)=
  \Xi(d, \tilde{x}_\lambda(d))$ and then using the chain rule, we have
  \begin{align}
    \nu_\lambda'(d) &= \Xi^{(1,0)}(d,\tilde{x}) +
    \Xi^{(0,1)}(d,\tilde{x})\diff{\tilde{x}}{d}\nonumber\\
    &=\Xi^{(1,0)}(d,\tilde{x}), \text{ since $\Xi^{(0,1)}(d,\tilde{x})
      = 0$ by definition of $\tilde{x}$}\nonumber\\
    &=\Xi(d,\tilde{x})\insq{\frac{q-1}{d} -
      \frac{q\log(1+\tilde{x})}{2(1+f_{d,\lambda}(\tilde{x}))}}\nonumber\\
    &=\frac{q\Xi(d, \tilde{x})}{d}\insq{1 - \frac{1}{q} -
      \frac{\log\inp{1+\tilde{x}}}{2\tilde{x}}}.\label{eq:18}
  \end{align}
  Here, we use $1 + f_{d,\lambda}(\tilde{x}) = d\tilde{x}$ to get the
  last equality.  We now note that the quantity inside the square
  brackets is an increasing function of $\tilde{x}$, and hence a
  decreasing function of $d$.  Since $\Xi(d, \tilde{x})$ is positive,
  this implies that there can be at most one positive zero of
  $\nu_\lambda(d)$, and if such a zero exists, it is the unique
  maximum of $\nu_\lambda(d)$.

  We now complete the proof by showing that $\nu_\lambda'(d) = 0$ for
  $d = \Delta_c(\lambda)$.  At such a $d$, we
  have $\lambda = \lambda_c(d) = \frac{d^d}{(d-1)^{d+1}}$.  We then
  observe that $\tilde{x}(d) = \frac{1}{d-1}$, since
  \begin{displaymath}
    1 + f_d(\tilde{x}) =
    1 + \frac{d^d}{(d-1)^{d+1}}\cdot\frac{(d-1)^d}{d^d} =
    \frac{d}{d-1} = d\tilde{x}.
  \end{displaymath}
  As an aside, we note that this is not a coincidence. Indeed, when
  $\lambda = \lambda_c(d)$, $\tilde{x}$ as defined above is well known
  to be the unique fixed point of $f_d$, and the potential function
  $\Phi$ was chosen in \cite{li_correlation_2011} in part to make sure
  that at the critical activity, the fixed point is also the
  maximizer of (an analog of) $\Xi(d, \cdot)$.

  We now substitute the value of $\frac{1}{q}$ and $\tilde{x}$ at $d =
  \Delta_c$ to verify that
  \begin{displaymath}
    \nu_\lambda(\Delta_c) = \frac{q\Xi\inp{\Delta_c,
        \frac{1}{\Delta_c - 1}}}{2\Delta_c}\insq{(\Delta_c -  1)\log\inp{1+ \frac{1}{\Delta_c -
          1}} - (\Delta_c - 1)\log\inp{1+\frac{1}{\Delta_c - 1}}} = 0,
  \end{displaymath}
  as claimed.  Substituting these values of $d$ and $\tilde{x}$, along
  with the earlier observation that $f_{\Delta_c}(\tilde{x}) =
  \tilde{x} = \frac{1}{\Delta_c - 1}$, into the definition of
  $\nu_\lambda$, we have
  \begin{align*}
    \nu_\lambda(\Delta_c) = \Xi\inp{\Delta_c, \frac{1}{\Delta_c - 1}}
    &= \Delta_c^{q-1} \inp{\frac{\tilde{x}}{1+\tilde{x}}
      \frac{f_{\Delta_c}(\tilde{x})}{1
        + f_{\Delta_c}(\tilde{x})}}^{q/2}\\
    &=\frac{1}{\Delta_c},
  \end{align*}
  which completes the proof.
\end{proof}

\subsection{Symmterizability of the message}
\label{sec:symmt-mess}
In this section, we prove Lemma~\ref{lem:observ-symm}.  We start with
the following technical lemma.
\begin{lemma}
  Let $r \geq 1$, $0 < A < 1$, $\gamma(x) \defeq (1-x)^r$ and $g(x)
  \defeq \gamma(Ax) + \gamma(A/x)$.  Note that $g(x) = g(1/x)$, and
  $g$ is well defined in the interval $[A, 1/A]$.  Then all the maxima
  of the function $g$ in the interval $[A, 1/A]$ lie in the set
  $\inb{1/A, 1, A}$. \label{lem:symm-basic}
\end{lemma}
Before proving the lemma, we observe the following simple
consequence. Consider $0 \leq s_1, s_2 \leq 1$ such that $s_1s_2$
is constrained to be some fixed constant $C < 1$.  Then, applying the
lemma with $A = \sqrt{C}$ we see that $\gamma(s_1)$ + $\gamma(s_2)$ is
maximized either when $s_1 = s_2$ or when one of them is $1$ and the
other is $C$.
\begin{proof}[Proof of Lemma~\ref{lem:symm-basic}]
  Note that when $r = 1$, $g(x) = 2 - A(x + 1/x)$, which is maximized
  at $x = 1$.  We therefore assume $r > 1$ in the following.

  We consider the derivative $g'(x)
  =Ar\insq{(1-A/x)^{r-1}\frac{1}{x^2} - (1-Ax)^{r-1}}$.  Note that $g(x) =
  g(1/x)$ and that $g'(x)$ and $g'(1/x)$ have opposite signs, so it is
  sufficient to study $g$ in the range $[1, 1/A]$.  We now note that
  in the interior of the intervals of interest $g'$ always has the
  same sign as
  $$h(x) \defeq A x^{t+1} - x^t + x - A,$$ where $t \defeq
  \frac{r+1}{r-1} > 1$ for $r>1$. We therefore only need to study the
  sign of $h$ in the interval $I \defeq [1, 1/A]$.  We note that $h(1)
  = 0$, and consider the derivatives of $h$.
  \begin{align*}
    h'(x) &= A(t+1)x^t - tx^{t-1} + 1,\\
    h''(x) &= t(t+1)x^{t-2}\insq{Ax - \frac{1}{r}}.
  \end{align*}
  Note that $h'(1) = (t+1)[A - 1/r]$.  We now break the analysis into
  two cases.
  \begin{description}
  \item[Case 1: $A \geq 1/r$.] In this case, we have $h''(x) > 0$ for
    $x$ in the interior of the interval $I$, and $h'(1) \geq 0$.  This
    shows that $h'(x) > 0$ for $x$ in the interior of $I$, so that $h$
    is strictly increasing in this interval.  Since $h(1) = 0$,
    this shows that $h$ (and hence $g'$) are positive in the interior
    of $I$.  Thus, $g$ is maximized in $I$ at $x = 1/A$.
  \item[Case 2: $A < 1/r$.] We now have $h'(1) < 0$ and $h''(1) < 0$.
    Further, defining $x_0 = \frac{1}{Ar}$, we see that $h''$ is
    negative in $[1, x_0)$ and positive in $(x_0, 1/A]$ (and $0$ at
    $x_0$).  Since $h'(1) < 0$, this shows that $h'$ is negative in
    $[1, x_0]$, and hence can have no zeroes there.  Further, we see
    that $h'$ is strictly increasing in $[x_0, 1/A]$, and hence can
    have at most one zero $x_1$ in $[x_0, 1/A]$.

    If no such zero exists, then $h'$ is negative in $I$.  In this
    case, we see that $h$ (and hence $g'$) is negative in the interior
    of $I$, and hence $g$ is maximized at $x = 1$.  We now consider
    the case where there is a zero $x_1$ of $h'$ in $[x_0, 1/A]$.  By
    the sign analysis of $h''$, we know that $h'$ is negative in $[1,
    x_1)$ and positive in $(x_1, 1]$.  We thus see that $h$ is
    decreasing (and negative) in $(1, x_1)$ and increasing in $(x_1,
    1/A)$.  It can therefore have at most one zero $x_2$ in $(x_1,
    1/A]$.  If no such zero exists, then $h$ (and hence $g'$) is
    negative in the interior of $I$, and hence $g$ is maximized at $x
    = 1$.  If such a zero $x_2$ exists in $(x_1, 1/A]$, then---because
    $h$ is increasing in $(x_1, 1/A)$ and negative in $(1, x_1]$---$h$
    (and hence $g'$) is negative in $(1, x_2)$ and positive in $(x_2,
    1/A)$, which shows that $g$ is maximized at either $x = 1$ or at
    $x = 1/A$.
  \end{description}
\end{proof}

We now prove Lemma~\ref{lem:observ-symm}.
\begin{proof}[Proof of Lemma~\ref{lem:observ-symm}]
  We first verify the second condition in the definition of
  symmetrizability:
  \begin{displaymath}
    \lim_{x\rightarrow 0^+} \frac{1}{(1+x)\Phi(x)} =
    \lim_{x\rightarrow 0^+} 2\sqrt{\frac{x}{1+x}} = 0.
  \end{displaymath}
  We now recall the program used in the definition of
  symmetrizability, with the definitions of $\Phi$ and $f_d$ substituted, and
  with $r = a/2$:
  \begin{align}
    \max\qquad &\sum_{i=1}^d \inp{\frac{x}{1+x}}^{r},\qquad \text{where}\nonumber\\
    &\lambda\prod_{i=1}^d\frac{1}{1+x_i} = B\label{eq:1}\\
    &x_i \geq 0,\qquad 1\leq i\leq d\nonumber
  \end{align}
  Note that eq. (\ref{eq:1}) implies that $x_i \leq \lambda/B - 1$, so
  that the feasible set is compact.  Thus, if the feasible set is
  non-empty, there is at least one (finite) optimal solution to the
  program.  Let $\vec{y}$ be such a solution.  Suppose without loss of
  generality that the first $k$ co-ordinates of $\vec{y}$ are non-zero
  while the rest are $0$.  We claim that $y_i = y_j \neq 0$ for all $1
  \leq i \leq j$ and $y_i = 0$ for $i > k$.

  To show this, we first define another vector $\vec{s}$ by setting
  $s_i = \frac{1}{1+ x_i}$.  Note that $s_i = s_j$ if and only if $x_i
  = x_j$ and $s_i = 1$ if and only if $x_i = 0$.  Note that the
  constraint in eq. (\ref{eq:1}) is equivalent to
  \begin{equation}
    \prod_{i=1}^ds_i = B/\lambda. \label{eq:3}
  \end{equation}
  Now suppose that there exist $i\neq j$ such that $y_iy_j \neq 0$ and
  $y_i \neq y_j$.  We then have $s_i \neq s_j$ and $0 < s_1,
  s_2 < 1$.  Now, since $r = a/2 \geq 1$ when $a \geq 2$,
  Lemma~\ref{lem:symm-basic} implies that at least one of the
  following two operations, performed while keeping the product
  $s_is_j$ fixed (so that the constraints in
  eqs. (\ref{eq:1},\ref{eq:3}) are satisfied), will increase the value of
  the sum $\gamma(s_i) + \gamma(s_j) =
  \inp{\frac{y_i}{1+y_i}}^r + \inp{\frac{y_j}{1+y_j}}^r$:
  \begin{enumerate}
  \item Making $s_i = s_j$, or
  \item Making $y_i = 0$ (so that $s_i = 1$).
  \end{enumerate}
  Thus, if $\vec{y}$ does not have all its non-zero entries equal, we
  can increase the value of the objective function while maintaining
  all the constraints.  This contradicts the fact that $\vec{y}$ is a
  maximum, and completes the proof.
\end{proof}

\section{Proofs omitted from Section~\ref{sec:spec-mess-monom}}

\subsection{Symmetrizability of the message}
\label{sec:symm-mess-monomer}

In this section, we prove Lemma~\ref{lem:observ-symm-monomer}.  As in
the case of the hard core model, we begin with an auxiliary technical
lemma.
\begin{lemma}
  Let $r$ and $a$ satisfy $1 < r \leq 2$ and $0 < a < 1$
  respectively. Consider the functions $\gamma(x) \defeq x^r(2-x)^r$
  and $g(x) \defeq \gamma(a-x) + \gamma(a+x)$.  Note that $g$ is even
  and is well defined in the interval $[-A, A]$, where $A \defeq
  \min(a, 1-a)$.  Then all the maxima of the function $g$ in the
  interval $[-A, A]$ lie in the set $\inb{-a, 0,
    a}$. \label{lem:symm-basic-monomer}
\end{lemma}
The lemma has the following simple consequence.  Let $ 0 \leq
s_1,s_2\leq 1$ be such that $(s_1 + s_2)/2$ is constrained to be some
fixed constant $a \leq 1$.  Then, applying the lemma with $s_1 = a-x,
s_2 = a + x$, we see that $\gamma(s_1) + \gamma(s_2)$ is maximized
either when $s_1 = s_2 = a$ or when one of them is $0$ and the other
is $2a$ (the second case can occur only when $a \leq 1/2$).

\begin{proof}[Proof of Lemma~\ref{lem:symm-basic-monomer}]
  Since $g$ is even, we only need to analyze it in the interval $[0,
  A]$, and show that restricted to this interval, its maxima lie in
  $\inb{0, a}$.

  We begin with an analysis of the third derivative of $\gamma$, which
  is given by
  \begin{equation}
    \gamma'''(x) = -4 r (r-1) (1-x)(1 - (1-x)^2)^{r-2}
    \insq{
      \frac{
        3-(2r -1)(1-x)^2
      }{
        1 - (1-x)^2
      }
    }.\label{eq:9}
  \end{equation}
  Our first claim is that $\gamma'''$ is strictly increasing in the
  interval $[0, 1]$ when $1 < r \leq 2$.  In the case when $r = 2$,
  the last two factors in eq. (\ref{eq:9}) simplify to constants, so that
  $\gamma'''(x) = -12r(r-1)(1-x)$, which is clearly strictly
  increasing.  When $1 < r < 2$, the easiest way to prove the claim is
  to notice that each of the factors in the product on the right hand
  side of is a strictly increasing non-negative function of $y = 1- x$
  when $x \in [0, 1]$ (the fact that the second and third factors are
  increasing and non-negative requires the condition that $r < 2
  $). Thus, because of the negative sign, $\gamma'''$ itself is a
  strictly \emph{decreasing} function of $y$, and hence a strictly
  \emph{increasing} function of $x$ in that interval.

  We can now analyze the behavior of $g$ in the interval $[0, A]$. We
  first show that when $a > 1/2$, so that $A = 1 - a \neq a$, $g$ does
  not have a maximum at $x = A$ when restricted to $[0,A]$.  We will
  achieve this by showing that when $1 > a > 1/2$, $g'(1-a) < 0$. To
  see this, we first compute $\gamma'(x) =
  2rx^{r-1}(2-x)^{r-1}(1-x)$, and then observe that
  \begin{align*}
    g'(1-a) &= \gamma'(1) - \gamma'(2a - 1)\\
    &= -\gamma'(2a-1) < 0\text{, since $0 < 2a -1 < 1$.}
  \end{align*}

  We now start with the observation that $g'''(x) = \gamma'''(a + x) -
  \gamma'''(a-x)$, so that because of the strict monotonicity of
  $\gamma'''$ in $[0,1]$ (which contains the interval $[0, A]$), we
  have $g'''(x) > 0$ for $x \in (0, A]$.  We note that this implies
  that $g''(x)$ is strictly increasing in the interval $[0, A]$. We
  also note that $g'(0) = 0$.  We now consider two cases.

  \begin{description}
  \item[Case 1: $g''(0) \geq 0$] Using the fact that $g''(x)$ is
    strictly increasing in the interval $[0, A]$ we see that $g''(x)$
    is also positive in the interval $(0, A]$ in this case.  This, along with
    the fact that $g'(0) = 0$, implies that $g'(x) > 0$ for $x \in (0,
    A]$, so that $g$ is strictly increasing in $[0, A]$ and hence is
    maximized only at $x = A$.  As proved above, this implies that the
    maximum of $g$ must be attained at $x = a$ (in other words, the
    case $g''(0) \geq 0$ cannot arise when $a > 1/2$ so that $A = 1 -a
    \neq a$).

  \item[Case 2: $g''(x) < 0$] Again, using the fact that $g''(x)$ is
    strictly increasing in $[0, A]$, we see that there is at most one
    zero $c$ of $g''$ in $[0, A]$.  If no such zero exists, then $g''$
    is negative in $[0, A]$, so that $g'$ is strictly decreasing in
    $[0, A]$.  Since $g'(0) = 0$, this implies that $g'$ is also
    negative in $(0, A)$ so that the unique maximum of $g$ in $[0, A]$
    is attained at $x = 0$.

    Now suppose that $g''$ has a zero $c$ in $(0, A]$.  As before, we
    can conclude that $g'$ is strictly negative in $[0, c]$, and
    strictly increasing in $[c, A]$.  Thus, if $g'(A) < 0$, $g'$ must
    be negative in all of $(0, A]$, so that $g$ is again maximized at
    $x = 0$ as in Case 1.  The only remaining case is when there
    exists a number $c_1 \in (c, A]$ such that $g'$ is negative in
    $(0, c_1)$ and positive in $(c_1, A]$.  In this case, we note that
    $g'(A) \geq 0$, so that---as observed above--we cannot have $A
    \neq a$.  Further, the maximum of $g$ in this case is at $x = 0$
    if $g(0) > g(A)$, and at $x = A$ otherwise.  Since we already
    argued that $A$ must be equal to $a$ in this case, this shows that
    the maxima of $g$ in $[0, A]$ again lie in the set $\inb{0,
      a}$. \qedhere
 \end{description}
\end{proof}

We now prove Lemma~\ref{lem:observ-symm-monomer}.
\begin{proof}[Proof of Lemma~\ref{lem:observ-symm-monomer}]
  We first verify the second condition in the definition of
  symmetrizability:
  \begin{displaymath}
    \lim_{p_i\rightarrow 0}
    \frac{1}{\Phi\inp{p_i}}
    \abs{\pdiff{f_{d,\gamma}}{p_i}}
    = \lim_{p_i\rightarrow 0}
    \frac{\gamma p_i(2-p_i)}{
      \inp{1 + \gamma\sum_{j=1}^d p_j}^2
    } = 0.
  \end{displaymath}
  We now recall the program used in the definition of symmetrizability
  with respect to exponent $r$, with the definitions of $\Phi$ and
  $f_{d,\gamma}$ substituted:
  \begin{align*}
    \max\qquad &\gamma^r f_{d,\gamma}(\vec{p})^{2r} \sum_{i=1}^d
    p_i^r(2-p_i)^r, \qquad \text{where}\\
    &\frac{1}{1 + \gamma \sum_{i=1}^d p_i} = B\\
    &0 \leq p_i \leq 1,\qquad 1\leq i\leq d
  \end{align*}
  Since we are only interested in the values of $\vec{p}$ solving the
  program, we can simplify the program as follows:
  \begin{align}
    \max\qquad &\sum_{i=1}^d p_i^r(2-p_i)^r,
    \qquad \text{where}\nonumber\\
    &\sum_{i=1}^d p_i = B' \defeq \frac{1-B}{\gamma B}
    \nonumber\\
    &0 \leq p_i \leq 1,\qquad 1\leq i\leq d\nonumber
  \end{align}
  We see that the feasible set is compact. Thus, if it is also
  non-empty, there is at least one (finite) optimal solution to the
  program.  Let $\vec{y}$ be such a solution.  Suppose without loss of
  generality that the first $k$ co-ordinates of $\vec{y}$ are non-zero
  while the rest are $0$.  We claim that $y_i = y_j \neq 0$ for all $1
  \leq i \leq j \leq k$.

  For if not, let $i\neq j$ be such that $y_iy_j \neq 0$ and $y_i \neq
  y_j$.  Let $y_i + y_j = 2a$. The discussion following
  Lemma~\ref{lem:symm-basic-monomer} implies that at least one of the
  following two operations, performed while keeping the sum $y_i +
  y_j$ fixed and ensuring that $y_i,y_j \in [0, 1]$ (so that all the
  constraints in the program are still satisfied), will increase the
  value of the sum $\gamma(y_i) + \gamma(y_j) = y_i^r(2-y_i)^r +
  y_j^r(2-y_j)^r$:
  \begin{enumerate}
  \item Making $y_i = y_j$, or
  \item Making $y_i = 0$ (so that $y_j = 2a$).  This case is possible
    only when $2a \leq 1$.
  \end{enumerate}
  Thus, if $\vec{y}$ does not have all its non-zero entries equal, we
  can increase the value of the objective function while maintaining
  all the constraints.  This contradicts the fact that $\vec{y}$ is a
  maximum, and completes the proof.
\end{proof}

\end{document}